%% file: root.tex
\documentclass{ws-ijprai}

\usepackage{overcite}

\input{preamble}

\usepgfplotslibrary{external} 
\begin{document}

\markboth{D.B. Blumenthal, J. Gamper, S. Bougleux \& L. Brun}{Upper Bounding the Graph Edit Distance Based on Rings and Machine Learning}

\catchline{}{}{}{DOI:}{10.1142/S0218001421510083}

\title{Upper Bounding the Graph Edit Distance Based on Rings and Machine Learning}

\author{David B. Blumenthal\footnote{Corresponding author.}}
\address{Technical University of Munich\\Chair of Experimental Bioinformatics\\Freising, Germany\\\email{david.blumenthal@wzw.tum.de}}

\author{Johann Gamper}
\address{Free University of Bozen-Bolzano\\Faculty of Computer Science\\Bozen, Italy\\\email{gamper@inf.unibz.it}}

\author{Sébastien Bougleux}
\address{Normandie Université\\ENSICAEN, UNICAEN, CNRS, GREYC\\Caen, France\\\email{sebastien.bougleux@unicaen.fr}}

\author{Luc Brun}
\address{Normandie Université\\ENSICAEN, UNICAEN, CNRS, GREYC\\Caen, France\\\email{luc.brun@ensicaen.fr}}

\maketitle

\begin{abstract}
The graph edit distance (\GED) is a flexible distance measure which is widely used for inexact graph matching. Since its exact computation is \NP-hard, heuristics are used in practice. A popular approach is to obtain upper bounds for \GED via transformations to the linear sum assignment problem with error-correction (\LSAPE). Typically, local structures and distances between them are employed for carrying out this transformation, but recently also machine learning techniques have been used. In this paper, we formally define a unifying framework \LSAPEGED for transformations from \GED to \LSAPE. We also introduce rings, a new kind of local structures designed for graphs where most information resides in the topology rather than in the node labels. Furthermore, we propose two new ring based heuristics \RING and \RINGML, which instantiate \LSAPEGED using the traditional and the machine learning based approach for transforming \GED to \LSAPE, respectively. Extensive experiments show that using rings for upper bounding \GED significantly improves the state of the art on datasets where most information resides in the graphs' topologies. This closes the gap between fast but rather inaccurate \LSAPE based heuristics and more accurate but significantly slower \GED algorithms based on local search.
\end{abstract}


\input{introduction}

\input{preliminaries}

\input{paradigms}

\input{rings}

\input{heuristics}

\input{experiments}

\input{conclusions}

\bibliographystyle{ws-ijprai}
\bibliography{root}

\vspace*{-0.01in}
\noindent
\rule{12.6cm}{.1mm}


\biophoto{blumenthal.jpg}{{\bf David B.\ Blumenthal} is a postdoctoral fellow at the Chair of Experimental Bioinformatics of the Technical University of Munich, Germany. He received the PhD in computer science from the Free  University of Bozen-Bolzano, Italy. His main interests are graph-based data management and algorithmic techniques for network and systems medicine.}

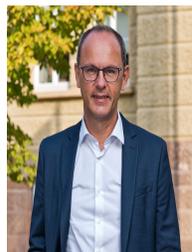
\vglue-2.00truein
\hspace*{2.45truein}
\biophoto{gamper.jpg}{{\bf Johann Gamper} is a full professor at the Faculty of Computer Science of the Free University of Bozen-Bolzano, Italy. He received the PhD in Computer Science from the RWTH Aachen, Germany. His main interests are database technologies for processing temporal and spatial data, data warehousing, and approximate query answering.}

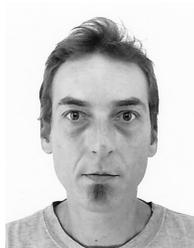
\biophoto{seb.jpg}{{\bf Sébastien Bougleux} is an associate professor at UNICAEN, ENSICAEN, CNRS, GREYC, Normandie Université, Caen, France. He received the PhD in computer science from the Université de Caen Normandie, France. His main interests are image processing and analysis as well as graph matching and graph similarity search.}

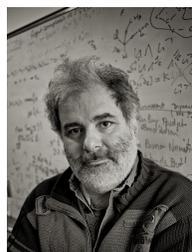
\vglue-2.00truein
\hspace*{2.45truein}
\biophoto{brun.jpg}{{\bf Luc Brun} is  a  full  professor  at  ENSICAEN, CNRS, GREYC, Normandie  Université, Caen, France. He received the PhD in computer science from the Université Bordeaux I, France.  His  main  research  interests  are  color image segmentation, combinatorial maps, pyramidal data structures, and metrics on graphs.}

\end{document}

%% file: preamble.tex
\usepackage{amsfonts,amssymb,amsmath,stmaryrd,mathtools,nicefrac,IEEEtrantools,url}
\usepackage{microtype}
\usepackage{breakurl}
\usepackage{hyperref}
\usepackage{cleveref}
\usepackage{csquotes}
\usepackage{xspace}
\usepackage{booktabs}
\usepackage{etoolbox,siunitx}
\usepackage{xcolor}
\usepackage{adjustbox,setspace}
\usepackage{tikz,pgf,pgfplots,pgfplotstable}
\usepgfplotslibrary{groupplots}
\pgfplotsset{compat=1.6}
\usetikzlibrary{calc,backgrounds,arrows,decorations.markings,decorations.pathreplacing,matrix,shapes,graphs,positioning,fit,snakes}
\usepackage[vlined,linesnumbered]{algorithm2e}
\makeatletter
\newcommand{\removelatexerror}{\let\@latex@error\@gobble}
\makeatother
\makeatletter
\newcommand{\gettikzxy}[3]{%
  \tikz@scan@one@point\pgfutil@firstofone#1\relax
  \edef#2{\the\pgf@x}%
  \edef#3{\the\pgf@y}%
}
\makeatother

\newcommand{\ie}{i.\,e.\@\xspace}

\newcommand{\wrt}{w.\,r.\,t.\@\xspace}
\newcommand{\cf}{cf.\@\xspace}

\newcommand{\mywlog}{w.\,l.\,o.\,g.\@\xspace}

\newcommand{\LSAPEGED}{\texttt{LSAPE-GED}\xspace}

\newcommand{\NP}{\ensuremath{\mathcal{NP}}\xspace}
\newcommand{\APX}{\ensuremath{\mathcal{APX}}\xspace}
\newcommand{\GI}{\ensuremath{\mathcal{GI}}\xspace}
\newcommand{\BP}{\texttt{BP}\xspace}

\newcommand{\STAR}{\texttt{STAR}\xspace}
\newcommand{\BRANCHUNI}{\texttt{BRANCH-UNI}\xspace}
\newcommand{\BRANCHFAST}{\texttt{BRANCH-FAST}\xspace}
\newcommand{\BRANCH}{\texttt{BRANCH}\xspace}
\newcommand{\WALKS}{\texttt{WALKS}\xspace}
\newcommand{\RING}{\texttt{RING}\xspace}
\newcommand{\RINGML}{\texttt{RING-ML}\xspace}
\newcommand{\RINGOPT}{\texttt{RING}\textsuperscript{\texttt{OPT}}\xspace}
\newcommand{\RINGGD}{\texttt{RING}\textsuperscript{\texttt{GD}}\xspace}
\newcommand{\RINGMS}{\texttt{RING}\textsuperscript{\texttt{MS}}\xspace}
\newcommand{\SUBGRAPH}{\texttt{SUBGRAPH}\xspace}

\newcommand{\IPFP}{\texttt{IPFP}\xspace}

\newcommand{\NGM}{\texttt{NGM}\xspace}
\newcommand{\PREDICT}{\texttt{PREDICT}\xspace}
\newcommand{\ALG}{\ensuremath{\mathtt{ALG}}\xspace}

\newcommand{\pah}{\textsc{pah}\xspace}
\newcommand{\alkane}{\textsc{alkane}\xspace}
\newcommand{\smol}[1]{\textsc{s\nobreakdash-mol\nobreakdash-\oldstylenums{#1}}}

\newcommand{\sacyclic}[1]{\textsc{s\nobreakdash-acyclic\nobreakdash-\oldstylenums{#1}}}
\newcommand{\smolnoarg}{\textsc{s\nobreakdash-mol}\xspace}

\newcommand{\sacyclicnoarg}{\textsc{s\nobreakdash-acyclic}\xspace}
\newcommand{\aids}{\textsc{aids}\xspace}

\newcommand{\acyclic}{\textsc{acyclic}\xspace}
\newcommand{\letter}{\textsc{letter}\xspace}

\newcommand{\GED}{\ifmmode\mathrm{GED}\else{GED}\fi\xspace}
\newcommand{\DNN}{DNN\xspace}
\newcommand{\SVC}{SVC\xspace}
\newcommand{\SVM}{1-SVM\xspace}
\newcommand{\LSAPE}{\ifmmode\mathrm{LSAPE}\else{LSAPE}\fi\xspace}
\newcommand{\LSAP}{\ifmmode\mathrm{LSAP}\else{LSAP}\fi\xspace}
\newcommand{\QAP}{\ifmmode\mathrm{QAP}\else{QAP}\fi\xspace}

\newcommand{\R}{\ensuremath{\mathbb{R}}\xspace}
\newcommand{\N}{\ensuremath{\mathbb{N}}\xspace}
\newcommand{\C}{\ensuremath{\mathbf{C}}\xspace}
\newcommand{\LAYER}{\ensuremath{\mathcal{L}}\xspace}
\newcommand{\RINGLS}[1][L]{\ensuremath{\mathcal{R}_{#1}}\xspace}
\newcommand{\NODES}{\ensuremath{\mathit{N}}\xspace}
\newcommand{\INNER}{\ensuremath{\mathit{IE}}\xspace}
\newcommand{\OUTER}{\ensuremath{\mathit{OE}}\xspace}

\newcommand{\defined}{\coloneqq}
\DeclareMathOperator*{\argmin}{arg\,min}

\DeclareMathOperator*{\supp}{supp}
\DeclareMathOperator*{\diam}{diam}

\newcommand{\norm}[1]{\left\lVert#1\right\rVert}

\tikzstyle{walks}=[mark=Mercedes star, mark size=3pt, color=blue, thick]
\tikzstyle{subgraph}=[mark=asterisk,mark size=3pt, color=teal, thick]
\tikzstyle{branch_fast}=[mark=10-pointed star,mark size=3pt, color=brown]
\tikzstyle{branch}=[mark=|,mark size=3pt, color=violet, thick]
\tikzstyle{bipartite}=[mark=-,mark size=3pt, color=black, thick]
\tikzstyle{ring_ms}=[mark=square, mark size=3pt, color=red, thick]
\tikzstyle{ring_gd}=[mark=o, mark size=3pt, color=red, thick]
\tikzstyle{ring_opt}=[mark=triangle, mark size=3pt, color=red, thick]
\tikzstyle{ring_ml_1svm}=[mark=x, mark size=3pt, color=magenta, thick]
\tikzstyle{bipartite_ml_1svm}=[mark=+,mark size=3pt, color=cyan, thick]
\tikzstyle{ring_ml_dnn}=[mark=otimes, mark size=3pt, color=magenta, thick]
\tikzstyle{bipartite_ml_dnn}=[mark=oplus,mark size=3pt, color=cyan, thick]
\tikzstyle{ipfp}=[mark=*, mark size=3pt, color=black, thick]

\newcommand{\addlogparetoplots}[6]{
\nextgroupplot[
ymode=log, 
title={#1}, 
legend to name=grouplegend, 
ylabel={#6},
xlabel={#5},
title style={yshift=-1ex},
xlabel style={yshift=1ex, font=\footnotesize},
ylabel style={yshift=-1ex, font=\footnotesize},
xticklabel style={font=\scriptsize},
yticklabel style={font=\scriptsize},
]
\addplot[ring_opt, only marks] table [y=#4_RING__LED-LSAPE_OPTIMAL, x=#3_RING__LED-LSAPE_OPTIMAL, col sep=comma, ignore chars={"}] {data/#2};
\addplot[ring_gd, only marks] table [y=#4_RING__LED-LSAPE_GREEDY, x=#3_RING__LED-LSAPE_GREEDY, col sep=comma, ignore chars={"}] {data/#2};
\addplot[ring_ms, only marks] table [y=#4_RING__LED-GAMMA, x=#3_RING__LED-GAMMA, col sep=comma, ignore chars={"}] {data/#2};
\addplot[ring_ml_1svm, only marks] table [y=#4_RING_ML__ML-ONE_CLASS_SVM_LIKELIHOOD, x=#3_RING_ML__ML-ONE_CLASS_SVM_LIKELIHOOD, col sep=comma, ignore chars={"}] {data/#2};
\addplot[ring_ml_dnn, only marks] table [y=#4_RING_ML__ML-DNN, x=#3_RING_ML__ML-DNN, col sep=comma, ignore chars={"}] {data/#2};
\addplot[bipartite_ml_1svm, only marks] table [y=#4_BIPARTITE_ML__ML-ONE_CLASS_SVM_LIKELIHOOD, x=#3_BIPARTITE_ML__ML-ONE_CLASS_SVM_LIKELIHOOD, col sep=comma, ignore chars={"}] {data/#2};
\addplot[bipartite_ml_dnn, only marks] table [y=#4_BIPARTITE_ML__ML-DNN, x=#3_BIPARTITE_ML__ML-DNN, col sep=comma, ignore chars={"}] {data/#2};
\addplot[walks, only marks] table [y=#4_WALKS, x=#3_WALKS, col sep=comma, ignore chars={"}] {data/#2};
\addplot[subgraph, only marks] table [y=#4_SUBGRAPH, x=#3_SUBGRAPH, col sep=comma, ignore chars={"}] {data/#2};
\addplot[branch, only marks] table [y=#4_BRANCH, x=#3_BRANCH, col sep=comma, ignore chars={"}] {data/#2};
\addplot[branch_fast, only marks] table [y=#4_BRANCH_FAST, x=#3_BRANCH_FAST, col sep=comma, ignore chars={"}] {data/#2};
\addplot[bipartite, only marks] table [y=#4_BIPARTITE, x=#3_BIPARTITE, col sep=comma, ignore chars={"}] {data/#2};
\addplot[ipfp, only marks] table [y=#4_IPFP, x=#3_IPFP, col sep=comma, ignore chars={"}] {data/#2};
\legend{\underline{\RINGOPT},\underline{\RINGGD},\underline{\RINGMS},\underline{\RINGML (\SVM)},\underline{\RINGML (\DNN)},\underline{\PREDICT (\SVM)},\PREDICT (\DNN),\WALKS,\SUBGRAPH,\BRANCH,\BRANCHFAST,\BP,\IPFP}
}

\newcommand{\addloglineplots}[6]{
\nextgroupplot[
ymode=log,
title={#1}, 
legend to name=grouplegend, 
ylabel={#6},
xlabel={#5},
title style={yshift=-1ex,},
xlabel style={yshift=1ex, font=\footnotesize},
ylabel style={yshift=-1ex, font=\footnotesize},
xticklabel style={font=\scriptsize},
yticklabel style={font=\scriptsize},
xtick={1,4,7,10},
]
\addlegendimage{ring_opt, only marks};
\addlegendimage{ring_gd, only marks};
\addlegendimage{ring_ms, only marks};
\addlegendimage{ring_ml_1svm, only marks};
\addlegendimage{ring_ml_dnn, only marks};
\addlegendimage{bipartite_ml_1svm, only marks};
\addlegendimage{bipartite_ml_dnn, only marks};
\addlegendimage{walks, only marks};
\addlegendimage{subgraph, only marks};
\addlegendimage{branch, only marks};
\addlegendimage{branch_fast, only marks};
\addlegendimage{bipartite, only marks};
\addplot[ring_opt] table [y=#4_RING__LED-LSAPE_OPTIMAL, x=#3, col sep=comma, ignore chars={"}] {data/#2};
\addplot[ring_gd] table [y=#4_RING__LED-LSAPE_GREEDY, x=#3, col sep=comma, ignore chars={"}] {data/#2};
\addplot[ring_ms] table [y=#4_RING__LED-GAMMA, x=#3, col sep=comma, ignore chars={"}] {data/#2};
\addplot[ring_ml_1svm] table [y=#4_RING_ML__ML-ONE_CLASS_SVM_LIKELIHOOD, x=#3, col sep=comma, ignore chars={"}] {data/#2};
\addplot[ring_ml_dnn] table [y=#4_RING_ML__ML-DNN, x=#3, col sep=comma, ignore chars={"}] {data/#2};
\addplot[bipartite_ml_1svm] table [y=#4_BIPARTITE_ML__ML-ONE_CLASS_SVM_LIKELIHOOD, x=#3, col sep=comma, ignore chars={"}] {data/#2};
\addplot[bipartite_ml_dnn] table [y=#4_BIPARTITE_ML__ML-DNN, x=#3, col sep=comma, ignore chars={"}] {data/#2};
\addplot[walks] table [y=#4_WALKS, x=#3, col sep=comma, ignore chars={"}] {data/#2};
\addplot[subgraph] table [y=#4_SUBGRAPH, x=#3, col sep=comma, ignore chars={"}] {data/#2};
\addplot[branch] table [y=#4_BRANCH, x=#3, col sep=comma, ignore chars={"}] {data/#2};
\addplot[branch_fast] table [y=#4_BRANCH_FAST, x=#3, col sep=comma, ignore chars={"}] {data/#2};
\addplot[bipartite] table [y=#4_BIPARTITE, x=#3, col sep=comma, ignore chars={"}] {data/#2};
\legend{\underline{\RINGOPT},\underline{\RINGGD},\underline{\RINGMS},\underline{\RINGML (\SVM)},\underline{\RINGML (\DNN)},\underline{\PREDICT (\SVM)},\PREDICT (\DNN),\WALKS,\SUBGRAPH,\BRANCH,\BRANCHFAST,\BP}
}

\newcommand{\addlineplots}[6]{
\nextgroupplot[
title={#1}, 
legend to name=grouplegend, 
ylabel={#6},
xlabel={#5},
title style={yshift=-1ex,},
xlabel style={yshift=1ex, font=\footnotesize},
ylabel style={yshift=-1ex, font=\footnotesize},
xticklabel style={font=\scriptsize},
yticklabel style={font=\scriptsize},
xtick={1,4,7,10},
]
\addlegendimage{ring_opt, only marks};
\addlegendimage{ring_gd, only marks};
\addlegendimage{ring_ms, only marks};
\addlegendimage{ring_ml_1svm, only marks};
\addlegendimage{ring_ml_dnn, only marks};
\addlegendimage{bipartite_ml_1svm, only marks};
\addlegendimage{bipartite_ml_dnn, only marks};
\addlegendimage{walks, only marks};
\addlegendimage{subgraph, only marks};
\addlegendimage{branch, only marks};
\addlegendimage{branch_fast, only marks};
\addlegendimage{bipartite, only marks};
\addplot[ring_opt] table [y=#4_RING__LED-LSAPE_OPTIMAL, x=#3, col sep=comma, ignore chars={"}] {data/#2};
\addplot[ring_gd] table [y=#4_RING__LED-LSAPE_GREEDY, x=#3, col sep=comma, ignore chars={"}] {data/#2};
\addplot[ring_ms] table [y=#4_RING__LED-GAMMA, x=#3, col sep=comma, ignore chars={"}] {data/#2};
\addplot[ring_ml_1svm] table [y=#4_RING_ML__ML-ONE_CLASS_SVM_LIKELIHOOD, x=#3, col sep=comma, ignore chars={"}] {data/#2};
\addplot[ring_ml_dnn] table [y=#4_RING_ML__ML-DNN, x=#3, col sep=comma, ignore chars={"}] {data/#2};
\addplot[bipartite_ml_1svm] table [y=#4_BIPARTITE_ML__ML-ONE_CLASS_SVM_LIKELIHOOD, x=#3, col sep=comma, ignore chars={"}] {data/#2};
\addplot[bipartite_ml_dnn] table [y=#4_BIPARTITE_ML__ML-DNN, x=#3, col sep=comma, ignore chars={"}] {data/#2};
\addplot[walks] table [y=#4_WALKS, x=#3, col sep=comma, ignore chars={"}] {data/#2};
\addplot[subgraph] table [y=#4_SUBGRAPH, x=#3, col sep=comma, ignore chars={"}] {data/#2};
\addplot[branch] table [y=#4_BRANCH, x=#3, col sep=comma, ignore chars={"}] {data/#2};
\addplot[branch_fast] table [y=#4_BRANCH_FAST, x=#3, col sep=comma, ignore chars={"}] {data/#2};
\addplot[bipartite] table [y=#4_BIPARTITE, x=#3, col sep=comma, ignore chars={"}] {data/#2};
\legend{\underline{\RINGOPT},\underline{\RINGGD},\underline{\RINGMS},\underline{\RINGML (\SVM)},\underline{\RINGML (\DNN)},\underline{\PREDICT (\SVM)},\PREDICT (\DNN),\WALKS,\SUBGRAPH,\BRANCH,\BRANCHFAST,\BP}
}

\newcommand{\addlineplotsnomlsmol}[6]{
\nextgroupplot[
title={#1}, 
legend to name=grouplegend, 
ylabel={#6},
xlabel={#5},
title style={yshift=-1ex,},
xlabel style={yshift=1ex, font=\footnotesize},
ylabel style={yshift=-1ex, font=\footnotesize},
xtick={1,4,7,10},
xticklabel style={font=\scriptsize},
yticklabel style={font=\scriptsize},
]
\addlegendimage{ring_opt, only marks};
\addlegendimage{ring_ms, only marks};
\addlegendimage{branch, only marks};
\addlegendimage{branch_fast, only marks};
\addlegendimage{bipartite, only marks};
\addplot[ring_opt] table [y=#4_RING__LED-LSAPE_OPTIMAL, x=#3, col sep=comma, ignore chars={"}] {data/#2};
\addplot[ring_ms] table [y=#4_RING__LED-GAMMA, x=#3, col sep=comma, ignore chars={"}] {data/#2};
\addplot[branch] table [y=#4_BRANCH, x=#3, col sep=comma, ignore chars={"}] {data/#2};
\addplot[branch_fast] table [y=#4_BRANCH_FAST, x=#3, col sep=comma, ignore chars={"}] {data/#2};
\addplot[bipartite] table [y=#4_BIPARTITE, x=#3, col sep=comma, ignore chars={"}] {data/#2};
\legend{\underline{\RINGOPT},\underline{\RINGMS},\BRANCH,\BRANCHFAST,\BP}
}

\newcommand{\addlineplotsnomlsmao}[6]{
\nextgroupplot[
title={#1}, 
legend to name=grouplegend, 
ylabel={#6},
xlabel={#5},
title style={yshift=-1ex,},
xlabel style={yshift=1ex, font=\footnotesize},
ylabel style={yshift=-1ex, font=\footnotesize},
xtick={3,5,7,9},
xticklabel style={font=\scriptsize},
yticklabel style={font=\scriptsize},
]
\addlegendimage{ring_opt, only marks};
\addlegendimage{ring_ms, only marks};
\addlegendimage{branch, only marks};
\addlegendimage{branch_fast, only marks};
\addlegendimage{bipartite, only marks};
\addplot[ring_opt] table [y=#4_RING__LED-LSAPE_OPTIMAL, x=#3, col sep=comma, ignore chars={"}] {data/#2};
\addplot[ring_ms] table [y=#4_RING__LED-GAMMA, x=#3, col sep=comma, ignore chars={"}] {data/#2};
\addplot[branch] table [y=#4_BRANCH, x=#3, col sep=comma, ignore chars={"}] {data/#2};
\addplot[branch_fast] table [y=#4_BRANCH_FAST, x=#3, col sep=comma, ignore chars={"}] {data/#2};
\addplot[bipartite] table [y=#4_BIPARTITE, x=#3, col sep=comma, ignore chars={"}] {data/#2};
\legend{\underline{\RINGOPT},\underline{\RINGMS},\BRANCH,\BRANCHFAST,\BP}
}

%% file: introduction.tex
\section{Introduction}\label{sec:intro}

Labeled graphs can be used for modeling various kinds of objects, such as chemical compounds, images, and molecular structures. Because of this flexibility, they have received increasing attention over the past years. One task researchers have focused on is the following: Given a database $\mathcal{G}$ that contains labeled graphs, find all graphs $G\in\mathcal{G}$ that are sufficiently similar to a query graph $H$ or find the $k$ graphs from $\mathcal{G}$ that are most similar to $H$ \cite{foggia:2014aa,vento:2015aa}. Being able to quickly answer graph similarity queries is crucial for the development of performant pattern recognition techniques in various application domains \cite{stauffer:2017aa}, such as keyword spotting in handwritten documents \cite{stauffer:2016aa} and cancer detection \cite{ozdemir:2013aa}. 

For answering graph similarity queries, a distance measure between two labeled graphs $G$ and $H$ has to be defined. A very flexible, sensitive and therefore widely used measure is the graph edit distance (\GED), which is defined as the minimum cost of an edit path between $G$ and $H$ \cite{bunke:1983aa}. An edit path is a sequence of graphs starting at $G$ and ending at a graph that is isomorphic to $H$ such that each graph on the path can be obtained from its predecessor by applying one of the following edit operations: adding or deleting an isolated node or an edge, and relabelling an existing node or edge. Each edit operation comes with an associated non-negative edit cost, and the cost of an edit path is defined as the sum of the costs of its edit operations. \GED inherits metric properties from the underlying edit costs \cite{justice:2006aa}. For instance, if $\mathfrak{G}$ is the domain of graphs with real-valued node and edge labels and the edit costs are defined as the Euclidean distances between the labels, then \GED is a metric on $\mathfrak{G}$.

\GED is mainly used in settings where we have to answer fine-grained similarity queries for (possibly very many) rather small graphs. For instance, this is the case in keyword spotting in handwritten documents, cancer  detection, and drug discovery \cite{stauffer:2017aa}. If the graphs are larger, faster approaches such as embedding the graphs into vector spaces and then comparing their vector representations may be used \cite{vento:2015aa,foggia:2014aa}. The drawback of these faster techniques is that a substantial part of the local information encoded in the original graphs is lost when embedding them into vector spaces. Whenever this information loss is intolerable, it is advisable to compare the graphs directly in the graph space\,---\,and one of the most commonly used distance measure for doing so is \GED.

\subsection{Related work}\label{sec:intro:related}

Computing \GED is a very difficult problem. It has been shown that the problem of computing \GED is \NP-hard even for uniform edit costs \cite{zeng:2009aa}, and \APX-hard for metric edit costs \cite{lin:1994aa}. Even worse: since, by definition of \GED, it holds that $\GED(G,H)=0$ just in case $G$ and $H$ are isomorphic, approximating \GED within any approximation ratio is \GI-hard. These theoretical complexities are mirrored by the fact that, in practice, no available exact algorithm can reliably compute \GED on graphs with more than 16 nodes \cite{blumenthal:2020ac}. 

Because of \GED's complexity, research has
mainly focused on heuristics \cite{blumenthal:2020aa}. The development of heuristics was
particularly triggered by the presentation of the algorithms \BP
\cite{riesen:2009aa} and \STAR \cite{zeng:2009aa}, which use
transformations to the linear sum assignment problem with error
correction (\LSAPE)\,---\,a variant of the linear sum assignment
problem (\LSAP) where rows and columns may also be inserted and
deleted\,---\,to compute upper bounds for \GED. \BP and \STAR work as
follows: First, the nodes of the input graphs are associated
with local structures composed of branches (\BP) or stars (\STAR),
respectively. The branch of a node is defined as the node itself together with its
incident edges, whereas the star additionally contains the incident edges' terminal
nodes.  Then, suitable distance measures
between, respectively, branches and stars are used to populate
instances of \LSAPE whose rows and columns correspond to the input graphs' nodes. Finally, the solution of the \LSAPE instance is
interpreted as an edit path whose cost is
returned as upper bound for \GED.

Following \BP and \STAR, many similar algorithms have been
proposed. Like \BP, the algorithms \BRANCHUNI \cite{zheng:2015aa},
\BRANCHFAST \cite{blumenthal:2018aa}, and \BRANCH
\cite{blumenthal:2018aa} use branches as local structures, but use
slightly different distance measures between the branches that allow
to also derive lower bounds.  The main disadvantage of these
algorithms as well as of \BP and \STAR is that, due to the very narrow
locality of the local structures, they yield unsatisfactorily loose
upper bounds on datasets where the nodes only carry little information 
and most information instead resides in the graphs' topologies. 

In order to produce tight upper bounds even if there is little information on the nodes, the algorithms \SUBGRAPH
\cite{carletti:2015aa} and \WALKS \cite{gauzere:2014aa} associate the
nodes with larger local structures\,---\,namely, subgraphs of fixed
radiuses and sets of walks of fixed lengths. The drawback of \SUBGRAPH
is that it runs in polynomial time only if the input
graphs have constantly bounded maximum degrees. \WALKS avoids this
blowup, but only models constant edit costs
and uses local structures that may contain redundant information due
to multiple inclusions of nodes and edges. It has also been
suggested to tighten the upper bounds by incorporating node centrality
measures into the \LSAPE instances \cite{riesen:2014ad,serratosa:2015aa}.

Recently, two machine learning based heuristics for \GED have been
proposed, which are closely related to the aforementioned \LSAPE based approaches. 
The algorithm \PREDICT \cite{riesen:2016aa} calls \BP to
compute an \LSAPE instance
and uses statistics of the \LSAPE instance to define feature vectors
for all node assignments. \PREDICT then trains a support vector
classifier (\SVC)
to predict if a node edit assignment is contained in an optimal edit
path. The algorithm \NGM
\cite{cortes:2018aa} defines feature vectors of node substitutions in
terms of the input graphs' node labels and node degrees. Given a set
of optimal edit paths, \NGM trains a deep neural network (\DNN) to
output a value close to $0$ if the substitution is predicted to be in
an optimal edit path, and close to $1$, otherwise. The output is used
to populate an instance of the linear sum assignment problem (\LSAP),
whose solution induces an upper bound for \GED. \NGM does not
support node insertions and deletions and can hence be used only for
equally sized graphs. Moreover, it ignores edge labels and assumes
that the node labels are real-valued vectors. \PREDICT works for
general graphs but does not yield an upper bound for \GED.

Not all heuristics for \GED build upon transformations to \LSAPE. For instance, further machine learning based algorithms have been suggested in Ref.~\refcite{li:2018aa,li:2019aa,bai:2019aa}. 
In Ref.~\refcite{li:2018aa}, Bayesian networks are employed to learn probabilities $\mathrm{Pr}(\GED(G,H)\leq\tau)$, which are then used to approximatively answer graph similarity queries.
In Ref.~\refcite{li:2019aa,bai:2019aa}, graph neural networks are used to learn \GED estimates from ground truth training data. All of these approaches are designed for
uniform edit costs only, \ie, for the case where all edit operations have the same cost. This assumption simplifies the concept of \GED, but is not met in most real-world application scenarios where \GED is used \cite{stauffer:2017aa}. Moreover, no assignments between the nodes of the compared graphs are computed, and the returned estimates for \GED are not guaranteed to be upper or lower bounds.

The tightest lower bounds for \GED are computed by heuristics based on linear programming \cite{justice:2006aa,lerouge:2016aa,lerouge:2017aa}. These algorithms formulate the problem of computing \GED as integer linear programs and call solvers such as CPLEX \cite{cplex:2016aa} or Gurobi \cite{gurobi-optimization:2018aa} to solve the continuous relaxations. Subsequently, the costs of the optimal continuous solutions are returned as lower bounds for \GED. For uniform edit costs, heuristics based on graph partitioning also yield good lower bounds \cite{zhao:2013aa,zheng:2015aa,liang:2017aa,zhao:2018aa,kim:2019ab}. In these approaches, one of the two input graphs is partitioned into small subgraphs. Subsequently, it is counted how many of the small subgraphs are not subgraph-isomorphic to the other input graph. If the edit costs are uniform, this count is a lower bound for \GED.

For upper bounding \GED, \LSAPE based heuristics are still among the most competitive approaches, as shown in a recent experimental evaluation survey \cite{blumenthal:2020aa}. In terms of accuracy, they are outperformed only by algorithms using variants of local search \cite{zeng:2009aa,riesen:2015ac,ferrer:2015ab,riesen:2017aa,bougleux:2017aa,daller:2018aa,boria:2018aa,blumenthal:2018ac,boria:2020aa}. However, these algorithms are much slower than \LSAPE based heuristics. Among local search based approaches, \IPFP is the best performing algorithm  \cite{bougleux:2017aa,daller:2018aa,boria:2018aa,blumenthal:2018ac,boria:2020aa}. In Ref.~\refcite{blumenthal:2020aa}, it is shown empirically that, on six benchmark datasets, \IPFP's upper bound exceeds the best available lower bounds, and hence the true \GED, by at most \SI{4.23}{\percent}. 

\subsection{Contributions}

In this paper, we formally describe a paradigm \LSAPEGED that generalizes all existing transformations from \GED to \LSAPE. Classical instantiations of \LSAPEGED such as \BP, \STAR, \BRANCHUNI, \BRANCHFAST, \BRANCH, \SUBGRAPH, and \WALKS use local structures and distance measures between them for the transformation. We introduce a partial order that compares classical instantiations in terms of the employed local structures' topological information content, and use it to systematically compare \BP, \STAR, \BRANCHUNI, \BRANCHFAST, \BRANCH, \SUBGRAPH, and \WALKS.

We also suggest
a new, machine learning based approach inspired by the algorithms
\PREDICT and \NGM: During training, feature vectors for all possible
node assignments are constructed and a machine learning
framework is trained to output a value close to $0$ if
a node assignment is predicted to be contained in an optimal edit path,
and a value close to $1$, otherwise. At runtime, the output of the
machine learning framework is fed into an \LSAPE instance. 

As mentioned above, \PREDICT and \NGM use \SVC{s} or \DNN{s} as their 
machine learning frameworks. They hence require training data that consists
of node assignments some of which are and some of which are not contained
in optimal edit paths. We argue that constructing such training data 
in a clean way is computationally infeasible. In order to overcome this problem, we suggest to use 
one class support vector machines (\SVM) instead of \SVC{s} and \DNN{s}.

Next, we present a new kind of local
structures\,---\,namely, rings of fixed sizes. Rings are sequences of disjoint node and edge sets at fixed distances from a root node. Like
the local structures used by \SUBGRAPH and \WALKS, rings are primarily designed for graphs where most information is encoded in the topologies rather than in the node labels. Using the partial order introduced before, we prove that rings indeed capture more topological information than the local
structures used by the baseline approaches \BP, \STAR, \BRANCHUNI, \BRANCHFAST, and \BRANCH. The
advantage of rings \wrt subgraphs is that rings can be processed in polynomial
time. The advantage \wrt walks is that rings model general edit costs
and avoid redundancies due to multiple inclusions of nodes and edges.

Subsequently, we propose \RING and \RINGML, two instantiations of \LSAPEGED
that make crucial use of rings. \RING adopts the classical approach,
\ie, carries out the transformation via a suitably defined ring
distance measure. In contrast to that, \RINGML uses rings to construct
feature vectors for the node assignments and then uses machine
learning techniques to carry out the transformation. 

An extensive
empirical evaluation shows that \RINGML shows very promising
potential and that, among all instantiations of
\LSAPEGED, \RING produces the tightest upper bound for \GED, especially on datasets where most information is encoded in the graphs's topologies. The newly proposed heuristic \RING hence closes the gap between existing instantiations of \LSAPEGED and the tighter but much slower local search algorithm \IPFP. In sum, our paper contains the following contributions:

\begin{itemize}
\item We formalize the paradigm \LSAPEGED and the
  classical, local structure distance based approach for transforming
  \GED to \LSAPE, and show how to use machine learning for
  this purpose (\Cref{sec:paradigm}).
\item We argue that \SVM{s} instead of classifiers such as
\DNN{s} or \SVC{s} should be used in machine learning based transformations to \LSAPE (\Cref{sec:paradigm}). 
\item We introduce rings, a new kind of local structures designed to yield tight upper bounds if most information is encoded in the graphs' topologies rather than in the node labels (\Cref{sec:rings}).
\item We present two new \LSAPEGED instantiations \RING and \RINGML (\Cref{sec:heuristics}).
\item We report results of extensive experiments (\Cref{sec:exp}).
\end{itemize}

This paper extends the results presented in Ref.~\refcite{blumenthal:2018ad},
where we informally described \LSAPEGED, introduced rings, and
presented and preliminarily evaluated \RING. In particular, the following contributions are new: formal definition of \LSAPEGED, the use of machine learning techniques in general and \SVM{s} in particular for transforming \GED to \LSAPE, the algorithm \RINGML, and additional experiments.

The rest of the paper is organized as follows:
In \Cref{sec:prelim}, we introduce concepts and notations.
In \Crefrange{sec:paradigm}{sec:exp}, we present our contributions. \Cref{sec:conc} concludes the paper.

%% file: preliminaries.tex
\section{Preliminaries}\label{sec:prelim}

We consider undirected labeled graphs $G=(V^G,E^G,\ell^G_V,\ell^G_E)$
from a domain of graphs $\mathfrak{G}$. $V^G$ and $E^G$ are sets of
nodes and edges, $\Sigma_V$ and $\Sigma_E$ are label alphabets, and
$\ell^G_V:V^G\to\Sigma_V$ and $\ell^G_E:E^G\to\Sigma_E$ are labeling
functions.
$\mathfrak{I}\defined \{(G,u)\mid G\in\mathfrak{G}\land u\in
(V^G\cup{\epsilon})\}$
is the set of all graph-node incidences and
$\mathfrak{A}\defined \{(G,H,u,v)\mid
(G,u)\in\mathfrak{I}\land(H,v)\in\mathfrak{I}\land(u\neq\epsilon\lor
v\neq\epsilon)\}$
is the set of all node assignments. The symbol $\epsilon$ denotes dummy
nodes and edges as well as their labels.  For each
$N\in\mathbb{N}_{\geq1}$, we define
$[N]\defined\{n\in\mathbb{N}_{\geq1}\mid n\leq N\}$.

An edit path from a graph $G$ to a graph $H$ is a sequence of edit operations that transforms $G$ into $H$. There are six edit operations: Substituting a node or and edge from $G$ by a node or an edge from $H$, deleting an isolated node or an edge from $G$, and inserting an isolated node or an edge between two existing nodes into $H$. Each edit operation $o$ comes with an edit cost $c(o)$ defined in terms of edit cost functions $c_V:\Sigma_V\cup\{\epsilon\}\times\Sigma_V\cup\{\epsilon\}\to\R_{\geq 0}$ and $c_E:\Sigma_E\cup\{\epsilon\}\times\Sigma_E\cup\{\epsilon\}\to\R_{\geq 0}$ (\cf \Cref{tab:edit-costs}), which respect $c_V(\alpha,\alpha)=0$ and $c_E(\beta,\beta)=0$ for all $\alpha\in\Sigma_V\cup\{\epsilon\}$ and all $\beta\in\Sigma_E\cup\{\epsilon\}$. The cost of an edit path $P=(o_i)^r_{i=1}$ is defined as $c(P)\defined \sum^r_{i=1}c(o_i)$.

\begin{table}[!t]
\centering
\caption{Notation table.}\label{tab:edit-costs}
\small
\input{tables/notation}
\end{table}

\begin{definition}[\boldmath \GED]\label{def:ged}
Let $\Psi(G,H)$ be the set of all edit paths from a graph $G$ to a graph $H$. Then graph edit distance is defined as $\GED(G,H)\defined \min_{P\in\Psi(G,H)}{c(P)}$.
\end{definition}

\Cref{def:ged} is very intuitive but algorithmically inaccessible, because for recognizing an edit path as such, one has to solve the graph isomorphism problem.
Thus, algorithms for \GED use an alternative definition based on
the concept of error-correcting matchings or node maps \cite{bunke:1983aa,bougleux:2017aa,blumenthal:2020ac}.

\begin{definition}[Node Map]\label{def:node-map}
Let $G$ and $H$ be graphs. A relation $\pi\subseteq (V^G\cup\{\epsilon\})\times(V^H\cup\{\epsilon\})$ is called node map between $G$ and $H$ if and only if $|\{v \mid v\in(V^H\cup\{\epsilon\})\land(u,v)\in\pi\}|=1$ holds for all $u\in V^G$ and $|\{u \mid u\in(V^G\cup\{\epsilon\})\land(u,v)\in\pi\}|=1$ holds for all $v\in V^H$. We write $\pi(u)=v$ just in case $(u,v)\in\pi$ and $u\neq\epsilon$, and $\pi^{-1}(v)=u$ just in case $(u,v)\in\pi$ and $v\neq\epsilon$. $\Pi(G,H)$ denotes the set of all node maps between $G$ and $H$.
\end{definition}

A node map $\pi\in\Pi(G,H)$ specifies for all nodes $u\in V^G$ and $v\in V^H$ and all edges $e=(u_1,u_2)\in E^G$ and $f=(v_1,v_2)\in E^H$ if they are substituted, deleted, or inserted: If $\pi(u)=v$, the node $u$ is substituted by $v$; if $\pi(u)=\epsilon$, $u$ is deleted; and if $\pi^{-1}(v)=\epsilon$, $v$ is inserted. Similarly, if $(\pi(u_1),\pi(u_2))=(v_1,v_2)$, the edge $e$ is substituted by $f$; if $(\pi(u_1),\pi(u_2))\notin E^H$, $e$ is deleted; and if $(\pi^{-1}(v_1),\pi^{-1}(v_2))\notin E^G$, $f$ is inserted. A node map $\pi\in\Pi(G,H)$ hence induces an edit path from $G$ to $H$. It has been shown that, for computing \GED, it suffices to consider edit paths induced by node maps \cite{bunke:1983aa,riesen:2015ad,bougleux:2017aa,blumenthal:2020ac}.

\begin{theorem}[\boldmath Alternative Definition of \GED]\label{thm:ged-node-map}
Let $G$ and $H$ be graphs and $c(\pi)$ be the cost of the edit path induced by a node map $\pi\in\Pi(G,H)$. Then, under mild constraints on the edit costs that can be assumed to hold \mywlog (\cf Ref.~\refcite{blumenthal:2020ac} for details), it holds that $\GED(G,H)=\min_{\pi\in\Pi(G,H)}{c(\pi)}$.
\end{theorem}

Since node maps are much easier objects to work with than edit paths, \Cref{thm:ged-node-map} renders \GED algorithmically accessible. In particular, it implies that each node map induces an upper
bound for \GED. This observation is used by approximative methods for
\GED, which heuristically compute a node map that induces a tight upper
bound. For finding such a node map, many heuristics \cite{riesen:2009aa,zeng:2009aa,zheng:2015aa,blumenthal:2018aa,carletti:2015aa,gauzere:2014aa,riesen:2014ad,serratosa:2015aa} use the linear sum
assignment problem with error correction (\LSAPE)
\cite{bougleux:2016aa,bougleux:2020aa}\,---\,although usually not under this name, since \LSAPE was formalized after the presentation of most of these heuristics.

\begin{definition}[LSAPE]\label{def:lsape}
Let  $\C=(c_{i,k})\in\mathbb{R}^{(n+1)\times(m+1)}$ be a matrix with $c_{n+1,m+1}=0$. A relation $\pi\subseteq[n+1]\times[m+1]$ is a feasible \LSAPE solution for \C, if and only if $|\{k \mid k\in[m+1]\land(i,k)\in\pi\}|=1$ holds for all $i\in[n]$ and $|\{i \mid i\in[n+1]\land(i,k)\in\pi\}|=1$ holds for all $k\in[m]$. We write $\pi(i)=k$ if $(i,k)\in\pi$ and $i\neq n+1$; and $\pi^{-1}(k)=i$ if $(i,k)\in\pi$ and $k\neq m+1$. The set of all feasible \LSAPE solutions for \C is denoted by $\Pi(\C)$. The cost of a feasible solution $\pi\in\Pi(\C)$ is defined as $\C(\pi)\defined \sum_{(i,k)\in\pi}c_{i,k}$. The set of all optimal \LSAPE solutions for \C is defined as $\Pi^\star(\C)\defined \argmin_{\pi\in\Pi(\C)}\C(\pi)$.
\end{definition}

An optimal \LSAPE solution $\pi\in\Pi^\star(\C)$ can be computed in
$\mathcal{O}(\min\{n,m\}^2\max\{n,m\})$ time \cite{bougleux:2020aa}. Once one
optimal solution has been found, for each $s\in[|\Pi^\star(\C)|]$, a
solution set $\Pi^\star_s(\C)\subseteq\Pi^\star(\C)$ of size $s$ can
be enumerated in $\mathcal{O}(nm\sqrt{n+m}+s\log{(n+m)})$ time
\cite{uno:2001aa}. Greedy suboptimal solutions can be computed in
$\mathcal{O}(nm)$ time \cite{riesen:2015aa}.

Node maps and feasible solutions for \LSAPE are closely
related. Consider graphs $G$ and $H$ and an \LSAPE instance
$\C\in\mathbb{R}^{(|V^G|+1)\times(|V^H|+1)}$. 
\Cref{def:node-map} and \Cref{def:lsape} imply that we can
identify the set $\Pi(G,H)$ of all node maps between $G$ and $H$ with
the set $\Pi(\C)$ of all feasible \LSAPE solutions for \C: For all
$i\in[|V^G|]$ and all $k\in[|V^H|]$, we associate \C's $i^\text{th}$
row with the node $u_i\in V^G$ and \C's $k^\text{th}$ column with the
node $v_k\in V^H$.  The last row and the last column of \C are
associated with the dummy node $\epsilon$. Then, each feasible
\LSAPE solution $\pi$ for \C yields an upper bound for
$\GED(G,H)$\,---\,namely, the cost $c(\pi)$ of the edit path induced
by $\pi$'s interpretation as a node map.

\begin{figure}[t]
\centering
\input{img/letters}
\caption{Two graphs $G$ and $H$ from the \letter dataset and a node map $\pi$.}\label{fig:letters}
\end{figure}
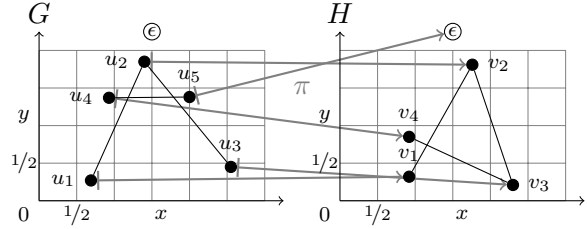

\begin{example}[Most Important Definitions]\label{expl:letters}
The graphs $G$ and $H$ shown in \Cref{fig:letters} are taken from the \letter dataset and represent distorted letter drawings \cite{riesen:2008aa}. Nodes are labeled with two-dimensional, non-negative Euclidean coordinates. Edges are unlabeled. Hence, we have $\Sigma_V=\R_{\geq0}^2$ and $\Sigma_E=\{1\}$. In Ref.~\refcite{riesen:2010aa}, the following edit cost functions are suggested: $c_E(1,\epsilon)\coloneqq c_E(\epsilon,1)\coloneqq0.425$, $c_V(\alpha,\alpha^\prime)\coloneqq0.75\norm{\alpha-\alpha^\prime}$, and $c_V(\alpha,\epsilon)\coloneqq c_V(\epsilon,\alpha)\coloneqq0.675$, where $\norm{\cdot}$ is the Euclidean norm. The node map $\pi\in\Pi(G,H)$ induces the following edit operations: substitutions of the nodes $u_i$ by the nodes $v_i$, for all $i\in[4]$, deletion of node $u_5$, substitutions of the edges $(u_i,u_{i+1})$ by the edges $(v_i,v_{i+1})$, for all $i\in[2]$, deletion of the edge $(u_4,u_5)$, insertion of the edge $(v_3,v_4)$. By summing the induced edit costs, we obtain that $\pi$'s induced edit path has cost $c(\pi)=2.574$ and hence $\GED(G,H)\leq2.574$.	
\end{example}

%% file: tables/notation.tex
\begin{tabular}{@{}ll@{}}
\toprule
syntax & semantics \\
\midrule
$\mathfrak{G}$ & graphs on label alphabets $\Sigma_V$ and $\Sigma_E$\\
$\mathfrak{I}$ & graph-node incidences in $\mathfrak{G}$\\
$\mathfrak{A}$ & node assignments between graphs in $\mathfrak{G}$\\
$c_V(\ell^G_V(u),\ell^H_V(v))$ & cost of substituting $u\in V^G$ by $v \in V^H$ \\
$c_E(\ell^G_E(e),\ell^H_E(f))$ & cost of substituting $e\in E^G$ by $f \in E^H$ \\
$c_V(\ell^G_V(u),\epsilon)$ & cost of deleting $u\in V^G$\\
$c_E(\ell^G_E(e),\epsilon)$ & cost of deleting $e\in E^G$ \\
$c_V(\epsilon,\ell^H_V(u))$ & cost of inserting $v \in V^H$\\
$c_E(\epsilon,\ell^H_E(f))$ & cost of inserting $f \in E^H$\\
$\Pi(G,H)$ & node maps between graphs $G$ and $H$\\
$c(\pi)$ & node map $\pi$'s induced edit cost\\
$\Pi(\C)$ & feasible solutions for \LSAPE instance \C\\
\bottomrule
\end{tabular}

%% file: img/letters.tex
\begin{tikzpicture}[every node/.style={font=\footnotesize}, scale=1, transform shape]
\tikzset{dot/.style={draw,fill=black,shape=circle,inner sep=1.5pt}}
\tikzset{eps/.style={draw,shape=circle,inner sep=.5pt}}
\node (G) at (0,2.5) {\large $G$};
\draw[step=.5cm,gray,very thin] (0,0) grid (3,2);
\node[dot,label=left:$u_1$] (u1) at (0.69,0.27) {};
\node[dot,label=left:$u_2$] (u2) at (1.40,1.85) {};
\node[dot,label=above:$u_3$] (u3) at (2.55,0.45) {};
\node[dot,label=left:$u_4$] (u4) at (0.93,1.37) {}; 
\node[dot,label=above:$u_5$] (u5) at (2.00,1.38) {};
\node[eps] (epsG) at (1.5,2.25) {$\epsilon$};
\draw (u1) -- (u2);
\draw (u2) -- (u3);
\draw (u4) -- (u5);
\draw[->] (0,0) -- node[below] (xG) {$x$} (3.25,0); 
\draw[->] (0,0) -- node[left] (yG) {$y$} (0,2.25);
\gettikzxy{(xG)}{\xGx}{\xGy};
\gettikzxy{(yG)}{\yGx}{\yGy};
\node (originG) at (\yGx,\xGy) {$0$};
\node at (0.5,\xGy) {$\nicefrac{1}{2}$};
\node at (\yGx,0.5) {$\nicefrac{1}{2}$};
\node (H) at (4,2.5) {\large $H$};
\draw[step=.5cm,gray,very thin] (4,0) grid (7,2);
\node[dot,label=above:$v_1$] (v1) at ($(0.92,0.32)+(4,0)$) {};
\node[dot,label=right:$v_2$] (v2) at ($(1.76,1.81)+(4,0)$) {};
\node[dot,label=right:$v_3$] (v3) at ($(2.30,0.21)+(4,0)$) {};
\node[dot,label=above:$v_4$] (v4) at ($(0.92,0.85)+(4,0)$) {};
\node[eps] (epsH) at (5.5,2.25) {$\epsilon$}; 
\draw (v1) -- (v2);
\draw (v2) -- (v3);
\draw (v3) -- (v4);
\draw[->] (4,0) -- node[below] (xH) {$x$} (7.25,0); 
\draw[->] (4,0) -- node[left] (yH) {$y$} (4,2.25);
\gettikzxy{(xH)}{\xHx}{\xHy};
\gettikzxy{(yH)}{\yHx}{\yHy};
\node (originH) at (\yHx,\xHy) {$0$};
\node (origin) at (4.5,\xHy) {$\nicefrac{1}{2}$};
\node (origin) at (\yHx,0.5) {$\nicefrac{1}{2}$};
\begin{scope}[on background layer]
\foreach \i in {1,...,4} {
\draw[|->,thick,gray] (u\i) -- (v\i);
}
\draw[|->,thick,gray] (u5) -- (epsH);
\node[gray] at (3.5,1.5) {\large $\pi$};
\end{scope}
\end{tikzpicture}

%% file: paradigms.tex
\section{\LSAPE Based Upper Bounds for \GED}\label{sec:paradigm}

In this section, we present the paradigm \LSAPEGED (\Cref{sec:paradigm:lsape-ged}). We 
show that existing LSAPE based approaches for upper bounding \GED are
instances of \LSAPEGED and define a partial order in terms of their local structures' topological information content (\Cref{sec:paradigm:classical}). Subsequently, 
we present a new machine learning technique to create \LSAPE
instances (\Cref{sec:paradigm:ml}). We also identify a problem that occurs if classifiers
such as \SVC{s} or \DNN{s} are used for creating the \LSAPE instances, and suggest that one should resort to \SVM{s} to overcome it.

\subsection{Overall structure of the paradigm \LSAPEGED}\label{sec:paradigm:lsape-ged}

\Cref{fig:lsape} shows how to use \LSAPE for upper bounding \GED:
Given two graphs $G$ and $H$, first an \LSAPE instance
$\C\in\mathbb{R}^{(|V^G|+1)\times(|V^H|+1)}$ is constructed. Different
strategies can be used for the construction, which are detailed in
\Cref{sec:paradigm:classical} and
\Cref{sec:paradigm:ml}. Subsequently, existing \LSAPE based heuristics
call a greedy or an optimal solver to compute an \LSAPE solution
$\pi\in\Pi(\C)$ and interpret $\pi$ as a node map whose induced
edit cost is returned as an upper bound for \GED. Alternatively, given
a constant $s>1$, we suggest to use an optimal \LSAPE solver in
combination with the enumeration procedure in Ref.~\refcite{uno:2001aa} in
order to generate a set $\Pi^\star_s(\C)\subseteq\Pi^\star(\C)$ of
optimal \LSAPE solutions with size
$|\Pi^\star_s(\C)|=\min\{s,|\Pi^\star(\C)|\}$. A tightened upper bound
for \GED can then be obtained by minimizing the induced edit cost over
the solution set $\Pi^\star_s(\C)$. Using the enumeration procedure
in Ref.~\refcite{uno:2001aa} to compute solution sets $\Pi^\star_s(\C)$ with
$s>1$ was suggested in Ref.~\refcite{daller:2018aa}, where $\Pi^\star_s(\C)$'s
elements are used as initial solutions for a refinement algorithm
based on local search. However, to the best of our knowledge, this
technique has never been used for tightening the upper bounds produced
by \LSAPE based heuristics.

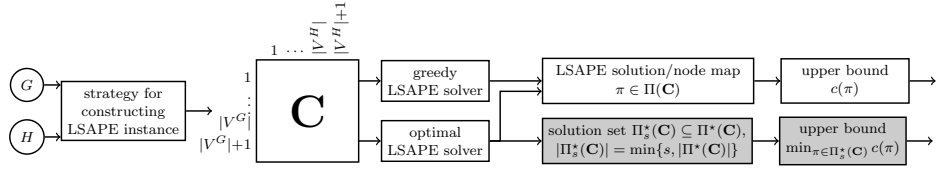
\begin{figure*}[t]
\centering
\begin{adjustbox}{max width=\linewidth}
\input{img/lsape}
\end{adjustbox}
\caption{The paradigm \LSAPEGED. Our alternative solution strategy is displayed in grey.}\label{fig:lsape}
\end{figure*}

\subsection{Classical instantiations}\label{sec:paradigm:classical}

Classical instantiations of the paradigm \LSAPEGED construct the
\LSAPE instance \C by using local structures rooted at the nodes and
distance measures between them.  Formally, they define local structure
functions $\mathcal{S}:\mathfrak{I}\to\mathfrak{S}$ that map
graph-node incidences to elements of a suitably defined space of
local structures $\mathfrak{S}$, and distance measures
$d_\mathfrak{S}:\mathfrak{S}\times\mathfrak{S}\to\mathbb{R}_{\geq0}$
for the local structures. Given input graphs $G$ and $H$ on node sets
$V^G=\{u_1,\ldots,u_{|V^G|}\}$ and $V^H=\{v_1,\ldots,v_{|V^H|}\}$, the
\LSAPE instance $\C\in\mathbb{R}^{(|V^G|+1)\times(|V^H|+1)}$ is then
defined as
$c_{i,k}\defined d_\mathfrak{S}(\mathcal{S}(G,u_i),\mathcal{S}(H,v_k))$, 
$c_{i,|V^H|+1}\defined d_\mathfrak{S}(\mathcal{S}(G,u_i),\mathcal{S}(H,\epsilon))$, and
$c_{|V^G|+1,k}\defined d_\mathfrak{S}(\mathcal{S}(G,\epsilon),\mathcal{S}(H,v_k))$,
for all $(i,k)\in[|V^G|]\times[|V^H|]$. This classical strategy for populating \C is adopted by \BP \cite{riesen:2009aa}, \STAR \cite{zeng:2009aa}, \BRANCHUNI \cite{zheng:2015aa}, \BRANCH \cite{blumenthal:2018aa}, \BRANCHFAST \cite{blumenthal:2018aa}, \WALKS \cite{gauzere:2014aa}, and \SUBGRAPH \cite{carletti:2015aa}, as well as by the algorithm \RING proposed in this paper (\Cref{sec:heuristics:ring}). Also the node centrality based heuristics \cite{riesen:2014ad,serratosa:2015aa} can be subsumed under this model; here, the \enquote{local structures} are simply the nodes' centralities.

\begin{figure}[t]
\centering
\input{img/lsape-ged}
\caption{\LSAPE instance \C for the graphs shown in \Cref{fig:letters} constructed by the toy instantiation of \LSAPEGED described in \Cref{expl:lsape-ged}. Bold-faced cells correspond to the node assignments contained in optimal solution.}\label{fig:lsape-ged}
\end{figure}

\begin{example}[Classical Instantiations of \LSAPEGED]\label{expl:lsape-ged}
Consider a very simple classical instantiation of \LSAPEGED that uses the input graphs' node labels as its local structures and the node edit costs $c_V$ as the local structure distances. \Cref{fig:lsape-ged} shows the obtained \LSAPE instance \C for the graphs $G$ and $H$ shown in \Cref{fig:letters} under the assumption that $c_V$ is defined as detailed in \Cref{expl:letters}. Its optimal solution $\pi\defined\{(i,i)\mid i\in[5]\}$ selects the bold-faced cells of \C and corresponds to the node map shown in \Cref{fig:letters}. Its \LSAPE cost is hence $\C(\pi)=1.774$, and the induced upper bound for \GED equals $c(\pi)=2.574$.
\end{example}

The formal specification of the paradigm \LSAPEGED allows us to introduce a partial order $\succeq_T$ that orders local structure functions employed by classical instantiations \wrt their topological information content (\cf \Cref{def:partial-order}). The intuition behind $\succeq_T$ is that a local structure function $\mathcal{S}_1$ is topologically more informative than another local structure function $\mathcal{S}_2$ if fixing the local structures $\mathcal{S}_1(G,u)$ implies that also the local structures $\mathcal{S}_2(G,u)$ remain unchanged.

\begin{definition}[Partial Order $\succeq_T$]\label{def:partial-order}
A local structure function $\mathcal{S}_1$ is \emph{weakly topologically more informative} than a local structure function $\mathcal{S}_2$ (in symbols: $\mathcal{S}_1\succeq_T\mathcal{S}_2$), if and only if $\mathcal{S}_1(G,u)=\mathcal{S}_1(H,v)$ implies $\mathcal{S}_2(G,u)=\mathcal{S}_2(H,v)$. $\mathcal{S}_1$ is \emph{strictly topologically more informative} than $\mathcal{S}_2$ ($\mathcal{S}_1\succ_T\mathcal{S}_2$), if and only if $\mathcal{S}_1\succeq_T\mathcal{S}_2$ and $\mathcal{S}_2\nsucceq_T\mathcal{S}_1$. 
\end{definition}

\Cref{prop:partial-order} states that $\succeq_T$ is indeed a partial order.

\begin{proposition}\label{prop:partial-order}
Let $\mathbb{S}$ be the set of all local structure functions. Then $\succeq_T$ is a partial order on $\mathbb{S}/\!{\sim_T}$, where the equivalence relation $\sim_T$ is defined as $\mathcal{S}_1\sim_T\mathcal{S}_2 :\Leftrightarrow\mathcal{S}_1\succeq_T\mathcal{S}_2\land\mathcal{S}_2\succeq_T\mathcal{S}_1$.
\end{proposition}

\begin{proof}
Reflexivity and anti-symmetry immediately follow from the definitions of $\succeq_T$ and $\sim_T$. For showing transitivity, let $\mathcal{S}_1,\mathcal{S}_2,\mathcal{S}_3\in\mathbb{S}$ be local structure functions with $\mathcal{S}_1\succeq_T\mathcal{S}_2\succeq_T\mathcal{S}_3$ and $(G,u),(H,v)\in\mathfrak{I}$ be graph-node incidences with $\mathcal{S}_1(G,u)=\mathcal{S}_1(H,v)$. From $\mathcal{S}_1\succeq_T\mathcal{S}_2$, we obtain $\mathcal{S}_2(G,u)=\mathcal{S}_2(H,v)$. Therefore, $\mathcal{S}_2\succeq_T\mathcal{S}_3$ implies $\mathcal{S}_3(G,u)=\mathcal{S}_3(H,v)$ and hence $\mathcal{S}_1\succeq_T\mathcal{S}_3$, as required.
\end{proof}

\Cref{prop:poset-soa} below orders existing classical \LSAPEGED instantiations \wrt the topological information content of the employed local structure functions. \BRANCHUNI, \BRANCHFAST, \BRANCH, and \BP are topologically equivalent, because they use the same local structure function $\mathcal{S}$ and differ only \wrt the distance measure $d_\mathfrak{S}$ used on top of it. Among all existing instantiations, \SUBGRAPH's local structure function captures most topological information. However, this comes at the price of having to use a distance measure $d_\mathfrak{S}$ that is not polynomially computable. 

\begin{figure}
\centering
\input{img/poset-soa}
\caption{Hasse diagram induced by strict partial order $\succ_T$ for local structure functions employed by existing instantiations of \LSAPEGED.}\label{fig:poset-soa}
\end{figure}
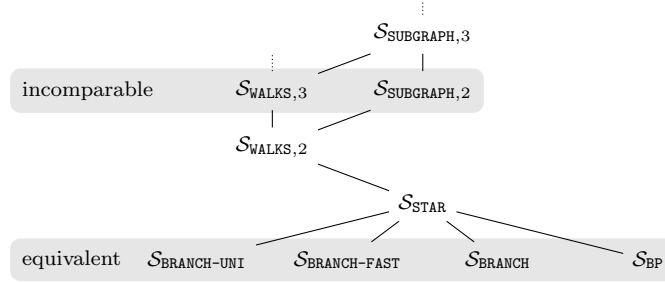

\begin{proposition}\label{prop:poset-soa}
The strict partial order $\succ_T$ orders local structure functions employed by existing classical \LSAPEGED instantiations as shown in the Hasse diagram in \Cref{fig:poset-soa}.
\end{proposition}

\begin{proof}
\BRANCHUNI, \BRANCHFAST, \BRANCH, and \BP all use the same local structure function\,---\,namely, branches rooted at the nodes. This implies $\mathcal{S}_\BRANCHUNI\sim_T\mathcal{S}_\BRANCHFAST\sim_T\mathcal{S}_\BRANCH\sim_T\mathcal{S}_\BP$. The star structures employed by $\mathcal{S}_\STAR$ can be viewed as branches that additionally contain the terminal nodes of the edges incident with the root or, alternatively, as the bag of all walks of length $1$ starting at the root. This yields $\mathcal{S}_{\WALKS,2}\succ_T\mathcal{S}_\STAR\succ_T\mathcal{S}_\BRANCHUNI,\mathcal{S}_\BRANCHFAST,\mathcal{S}_\BRANCH,\mathcal{S}_\BP$. The relations $\mathcal{S}_{\SUBGRAPH,L}\succ_T\mathcal{S}_{\WALKS,L}$, $\mathcal{S}_{\WALKS,L+1}\succ_T\mathcal{S}_{\WALKS,L}$,  as well as the fact that $\mathcal{S}_{\SUBGRAPH,L}$ and $\mathcal{S}_{\WALKS,L+1}$ are incomparable directly follow from the definitions of the respective local structure functions.
\end{proof}

\subsection{Machine learning based instantiations}\label{sec:paradigm:ml}

\subsubsection{General strategy}\label{sec:paradigm:ml:scheme}

Instead of using local structures, \LSAPE instances can be constructed
with the help of feature vectors associated to good and bad node
assignments. This strategy is inspired by the existing algorithms
\PREDICT \cite{riesen:2016aa} and \NGM \cite{cortes:2018aa}. However,
as detailed in \Cref{sec:heuristics:ngm-predict} below, both \PREDICT and \NGM fall short of completely
instantiating it.

\begin{definition}[$\epsilon$-Optimal Node Assignments]\label{def:good-bad}
  A node assignment $(G,H,u,v)\in\mathfrak{A}$ is called $\epsilon$-optimal if and
  only if it is contained in an $\epsilon$-optimal node map, \ie, if there is a
  node map $\pi\in\Pi(G,H)$ with $c(\pi)\leq\GED(G,H)\cdot(1+\epsilon)$ and
  $(u,v)\in\pi$. $\mathfrak{A}^\star_\epsilon$ is the set of all $\epsilon$-optimal node assignments. Node assignments contained in
  $\mathfrak{A}\setminus\mathfrak{A}^\star_\epsilon$ are called $\epsilon$-bad.
\end{definition}

If machine learning techniques are used for populating \C, feature vectors $\mathcal{F}:\mathfrak{A}\to\mathbb{R}^d$ for
the node assignments have to be defined and a function
$p^\star:\mathbb{R}^d\to[0,1]$ has to be learned, which maps feature
vectors $\mathbf{x}\in\mathcal{F}[\mathfrak{A}^\star_\epsilon]$ to large
and feature vectors
$\mathbf{x}\in\mathcal{F}[\mathfrak{A}\setminus\mathfrak{A}^\star_\epsilon]$ to
small values. Informally, $p^\star(\mathbf{x})$ can be viewed as an
estimate of the probability that the feature vector
$\mathbf{x}$ is associated to an $\epsilon$-optimal node
map. Once $p^\star$ has been learned, \C is defined as
$c_{i,k}\defined 1-p^\star(\mathcal{F}(G,H,u_i,v_k))$, 
$c_{i,|V^H|+1} \defined 1-p^\star(\mathcal{F}(G,H,u_i,\epsilon))$, and
$c_{|V^G|+1,k} \defined 1-p^\star(\mathcal{F}(G,H,\epsilon,v_k))$, 
for all $(i,k)\in[|V^G|]\times[|V^H|]$.

\subsubsection{Choice of machine learning technique}\label{sec:paradigm:ml:technique}

For learning $p^\star$, several strategies can be adopted. Given a set $\mathcal{G}$ of training graphs, one can mimic \PREDICT and compute $\epsilon$-optimal node maps $\pi_{G,H}$ for the training graphs. These node maps can be used to generate training data 
$\mathcal{T}\defined \{(\mathcal{F}(G,H,u,v),\delta_{(u,v)\in\pi_{G,H}})\mid(G,H,u,v)\in\mathfrak{A}[\mathcal{G}]\}$,
where $\delta_{\mathtt{true|false}}$ maps $\mathtt{true}$ to $1$ and $\mathtt{false}$ to $0$, and $\mathfrak{A}[\mathcal{G}]$ is the restriction of $\mathfrak{A}$ to the graphs contained in $\mathcal{G}$. Finally, a kernelized \SVC with probability estimates \cite{lin:2007aa} can be trained on $\mathcal{T}$. Alternatively, one can proceed like \NGM, \ie, use $\mathcal{T}$ to train a fully connected feedforward \DNN with output from $[0,1]$, and define  $p^\star$ as the output of the \DNN.

The drawback of these approaches is that some feature vectors
are incorrectly labeled as $\epsilon$-bad if there is
more than one $\epsilon$-optimal node map. Assume that, for training
graphs $G$ and $H$, there are two $\epsilon$-optimal node maps $\pi_{G,H}$ and
$\pi^\prime_{G,H}$ and that the algorithm used for generating
$\mathcal{T}$ computes $\pi_{G,H}$. Let $(G,H,u,v)$ be a node
assignment such that $(u,v)$ is contained in
$\pi^\prime_{G,H}\setminus\pi_{G,H}$. According to
\Cref{def:good-bad}, $(G,H,u,v)$ is an $\epsilon$-optimal node assignment, but in
$\mathcal{T}$, its feature vector
$\mathbf{x}\defined\mathcal{F}(G,H,u,v)$ is labeled as $\epsilon$-bad.

A straightforward but computationally infeasible way to tackle this
problem is to compute all $\epsilon$-optimal node maps for the training
graphs. Instead, we suggest to train a one class support vector
machine (\SVM) \cite{scholkopf:2001aa} with RBF kernel to estimate the
support of the feature vectors associated to good node maps. This has the advantage that, given a set
$\mathcal{G}$ of training graphs and initially computed $\epsilon$-optimal node
maps $\pi_{G,H}$ for all $G,H\in\mathcal{G}$, we can use training data
$\mathcal{T}^\star\defined \{\mathcal{F}(G,H,u,v)\mid(G,H,u,v)\in\mathfrak{A}[\mathcal{G}]\land(u,v)\in\pi_{G,H}\}$, 
which contains only feature vectors associated to $\epsilon$-optimal node assignments and is hence correct even if there are multiple $\epsilon$-optimal node maps.

For the definition of $p^\star$, recall that a \SVM with RBF kernel learns a multivariate Gaussian mixture model
$\mathcal{M}(\boldsymbol{\alpha},\gamma)\defined\sum_{i=1}^{|\mathcal{T}^\star|}(\mathbf{1}^\mathsf{T}\boldsymbol{\alpha})^{-1}\alpha_i\mathcal{N}(\mathbf{0},(2\gamma)^{-1}\mathbf{I})$
for the feature vectors $\mathcal{F}[\mathfrak{A}^\star_\epsilon]$ associated to $\epsilon$-optimal node assignments, where $\alpha_i$ is the dual variable associated to the training vector $\mathbf{x}^i\in\mathcal{T}^\star$ and $\gamma>0$ is a tuning parameter. We can hence simply define
$p^\star(\mathbf{x})$ as the likelihood of the feature vector
$\mathbf{x}$ under the model $\mathcal{M}(\boldsymbol{\alpha},\gamma)$
learned by the \SVM, \ie, set
$p^\star(\mathbf{x})\defined \left((\gamma/\pi)^{d/2}/\mathbf{1}^\mathsf{T}\boldsymbol{\alpha}\right)\cdot\left(\sum_{i=1}^{|\mathcal{T}^\star|}\alpha_i\exp{(-\gamma\lVert\mathbf{x}^i-\mathbf{x}\rVert^2_2)}\right)$.

%% file: img/lsape.tex
\tikzset{row_label/.style={anchor=east, xshift=-.25em}}
\tikzset{col_label_long/.style={anchor=west, yshift=.25em, rotate=90}}
\tikzset{col_label_short/.style={anchor=south, yshift=.25em}}
\begin{tikzpicture}[every node/.style={font=\footnotesize}, scale=1, transform shape]
\node[thick, circle, draw] (G) at (0,.5) {$G$};
\gettikzxy{(G)}{\Gx}{\Gy}
\node[thick, circle, draw] (H) at (0,-.5) {$H$};
\gettikzxy{(H)}{\Hx}{\Hy}
\coordinate (origin) at ($(G)!0.5!(H)$);
\node[thick, rectangle, draw, align=center, right = of origin,xshift=-1em] (constr) {strategy for \\ constructing \\ \LSAPE instance};
\gettikzxy{(constr.west)}{\constrx}{\constry}
\matrix (C) [
matrix of nodes,
nodes in empty cells,
row sep =-\pgflinewidth,
column sep = -\pgflinewidth,
nodes={minimum height=1.25em, minimum width=1.25em},
right=of constr,
xshift=1em
]
{ 
&     &          &         &           \\
&     &          &         &           \\
&     &          &         &           \\
&     &          &         &           \\
};
\node[row_label] (row_1) at (C-1-1.west) {$1$};
\node[row_label] (row_2) at (C-2-1.west) {$\vdots$};
\node[row_label] (row_3) at (C-3-1.west) {$|V^G|$};
\node[row_label] (row_4) at (C-4-1.west) {$|V^G|{+}1$};
\node[col_label_short] (col_1) at (C-1-1.north) {$1$};
\node[col_label_short] (col_2) at (C-1-2.north) {$\cdots$};
\node[col_label_long] (col_3) at (C-1-3.north) {$|V^H|$};
\node[col_label_long] (col_4) at (C-1-4.north) {$|V^H|{+}1$};
\node[thick, rectangle, draw, align=center, anchor=north west, minimum height=5ex] (greedy) at ($(C.north east)+(0,0)$) {greedy \\ \LSAPE solver};
\gettikzxy{(greedy)}{\greedyx}{\greedyy}
\node[thick, rectangle, draw, align=center, anchor=south west, minimum height=5ex] (optimal) at ($(C.south east)+(0,0)$) {optimal \\ \LSAPE solver};
\gettikzxy{(optimal.west)}{\optimalwestx}{\optimalwesty}
\gettikzxy{(optimal.east)}{\optimaleastx}{\optimaleasty}

\node[thick, draw, right=of optimal, align=center, minimum height=5ex, fill=black!20] (solutions) {solution set $\Pi^\star_s(\C)\subseteq\Pi^\star(\C)$, \\ $|\Pi^\star_s(\C)|=\min\{s,|\Pi^\star(\C)|\}$};

\node[thick, draw, align=center, minimum height=5ex, anchor=west, fill=black!20] (upper-bound-2) at ($(solutions.east)+(.5,0)$) {upper bound \\ $\min_{\pi\in\Pi^\star_s(\C)}c(\pi)$};

\coordinate (solvers) at ($(greedy.west)!.5!(optimal.west)$);
\gettikzxy{(solutions.west)}{\solutionswestx}{\solutionswesty}
\gettikzxy{(solutions.east)}{\solutionseastx}{\solutionseasty}

\gettikzxy{(upper-bound-2.east)}{\upperboundeastx}{\upperboundeasty}
\gettikzxy{(upper-bound-2.west)}{\upperboundwestx}{\upperboundwesty}

\node[thick, draw, align=center, anchor=west, minimum width=\solutionseastx-\solutionswestx, minimum height=5ex, fill=white] (solution) at ($(greedy.east)+(1,0)$) {\LSAPE solution/node map \\ $\pi\in\Pi(\C)$};

\node[thick, draw, align=center, anchor=west, minimum height=5ex,minimum width=\upperboundeastx-\upperboundwestx] (upper-bound) at ($(solution.east)+(.5,0)$) {upper bound \\ $c(\pi)$};

\gettikzxy{(solution)}{\solutionx}{\solutiony}

\gettikzxy{(upper-bound)}{\upperboundx}{\upperboundy}

\draw[thick, ->] (G) -- ($(\constrx,\Gy)+(0,0)$);
\draw[thick, ->] (H) -- ($(\constrx,\Hy)+(0,0)$);
\draw[->, thick] (constr) -- ($(constr.east)+(.5,0)$);
\draw[thick, ->] ($(optimal.west)+(-.75,0)$) -- (optimal.west);
\draw[thick, ->] ($(greedy.west)+(-.75,0)$) -- (greedy.west);
\draw[thick, ->] (optimal.east) -- (solutions.west);
\draw[thick, ->] (optimal.east) -- +(.2,0) |- ($(solution.west)+(0,-.2)$);
\draw[thick, ->] (greedy.east) -- (solution.west);
\draw[thick, ->] (solutions.east) -- (upper-bound-2.west);
\draw[thick, ->] (solution.east) -- (upper-bound.west);
\draw[thick, ->] (upper-bound.east) -- +(.5,0);
\draw[thick, ->] (upper-bound-2.east) -- +(.5,0);
\node[thick, fill=white, fit=(C-1-1) (C-4-4), draw] {};
\node at (barycentric cs:C-1-1=1,C-1-4=1,C-4-1=1,C-4-4=1) {\Huge \C};
\end{tikzpicture}

%% file: img/lsape-ged.tex
\tikzset{row_label/.style={anchor=east, xshift=-.5em, font=\footnotesize}}
\tikzset{col_label/.style={anchor=south, yshift=.5em, font=\footnotesize}}
\begin{tikzpicture}
\matrix (C) [
matrix of nodes,
nodes in empty cells,
nodes={font=\small,minimum width=3em}
]
{ 
$\mathbf{0.177}$ & $1.406$ & $1.208$ & $0.468$ & $0.675$ \\
$1.203$ & $\mathbf{0.272}$ & $1.403$ & $0.832$ & $0.675$ \\
$1.226$ & $1.180$ & $\mathbf{0.260}$ & $1.259$ & $0.675$ \\
$0.788$ & $0.705$ & $1.346$ & $\mathbf{0.390}$ & $0.675$ \\
$1.135$ & $0.369$ & $0.906$ & $0.902$ & $\mathbf{0.675}$ \\
$0.675$ & $0.675$ & $0.675$ & $0.675$ & $0$ \\
};
\node[row_label] (row_1) at (C-1-1.west) {$1$};
\node[row_label] (row_2) at (C-2-1.west) {$2$};
\node[row_label] (row_3) at (C-3-1.west) {$3$};
\node[row_label] (row_4) at (C-4-1.west) {$4$};
\node[row_label] (row_5) at (C-5-1.west) {$5$};
\node[row_label] (row_6) at (C-6-1.west) {$6$};
\node[col_label] (col_1) at (C-1-1.north) {$1$};
\node[col_label] (col_2) at (C-1-2.north) {$2$};
\node[col_label] (col_3) at (C-1-3.north) {$3$};
\node[col_label] (col_4) at (C-1-4.north) {$4$};
\node[col_label] (col_4) at (C-1-5.north) {$5$};
\node[thick, fit=(C-1-1) (C-6-5) (C-5-5), draw] {};
\draw[dashed] ($(C-5-1.west)!.5!(C-6-1.west)$) -- ($(C-5-5.east)!.5!(C-6-5.east)$);
\draw[dashed] ($(C-1-4.north)!.5!(C-1-5.north)$) -- ($(C-6-4.south)!.5!(C-6-5.south)$);
\end{tikzpicture}

%% file: img/poset-soa.tex
\begin{tikzpicture}[every node/.style={font=\footnotesize}, scale=1, transform shape]
\node (branchuni) at (-3,0) {$\mathcal{S}_\BRANCHUNI$};
\node (branchfast) at (-1,0) {$\mathcal{S}_\BRANCHFAST$};
\node (branch) at (1,0) {$\mathcal{S}_\BRANCH$};
\node (bp) at (3,0) {$\mathcal{S}_\BP$};
\node (star) at (0,.75) {$\mathcal{S}_\STAR$};
\node (walks2) at (-2,1.5) {$\mathcal{S}_{\WALKS,2}$};
\node (walks3) at (-2,2.25) {$\mathcal{S}_{\WALKS,3}$};
\node (subgraph2) at (0,2.25) {$\mathcal{S}_{\SUBGRAPH,2}$};
\node (subgraph3) at (0,3) {$\mathcal{S}_{\SUBGRAPH,3}$};

\node[anchor=west] (equiv) at (-5.45,0) {equivalent};
\node[anchor=west] (incomp3) at (-5.45,2.25) {incomparable};
\begin{scope}[on background layer]
\node[fit=(branchuni)(branchfast)(branch)(bp)(equiv),fill=black!10,inner sep=1pt, rounded corners] (box) {};
\node[fit=(subgraph2)(walks3)(incomp3),fill=black!10,inner sep=1pt, rounded corners] (box) {};
\end{scope}

\draw (branchuni) -- (star);
\draw (branchfast) -- (star);
\draw (branch) -- (star);
\draw (bp) -- (star);
\draw (star) -- (walks2);
\draw (walks2) -- (subgraph2);
\draw (walks2) -- (walks3);
\draw (walks3) -- (subgraph3);
\draw (subgraph2) -- (subgraph3);
\draw[densely dotted] (subgraph3) edge ($(subgraph3)+(0,.5)$);
\draw[densely dotted] (walks3) edge ($(walks3)+(0,.5)$);
\end{tikzpicture}

%% file: rings.tex
\section{Rings as Local Structures}\label{sec:rings}

In this section, we introduce rings of size $L$ as a new kind of local structures. Subsequently, we show how to construct them. As mentioned above, rings are similar to the subgraph and walks structures used, respectively, by \SUBGRAPH and \WALKS in that they capture more topological information than the local structures used by the baseline approaches \BP, \STAR, \BRANCHUNI, and \BRANCH. They are hence designed to be used by instantiations of \LSAPEGED that aim at computing tight upper bounds on datasets where most information resides in the graphs' topologies. The advantage of rings \wrt subgraphs is that rings can be processed in polynomial time, while comparing local subgraphs is polynomial only on graphs with constantly bounded maximum degrees. In comparison to walks, rings have two advantages: Firstly, rings model general edit costs, while walks are designed for constant edit costs. Secondly, walks often contain redundant information, because it can happen that a walk visits nodes and edges multiple times. Since rings consist of disjoint layers, they do not suffer from this problem.

\subsection{Rings: definition and properties}\label{sec:rings:definition}

We define the rings rooted at the nodes of a graph $G$ as $L$-sized sequences of layers $\mathcal{L}^G=(\NODES^G,\OUTER^G,\INNER^G)$, where $\NODES^G\subseteq V^G$ is a subset of the nodes, and $\OUTER^G,\INNER^G\subseteq E^G$ are subsets of the edges of $G$. Formally, the space of all $L$-sized rings for graphs from a domain $\mathfrak{G}$ is defined as $\mathfrak{R}_L:=\{(\mathcal{L}_l)^{L-1}_{l=0}\mid\mathcal{L}_l\in\bigcup_{G\in\mathfrak{G}}\mathfrak{L}(G)\}$, where $\mathfrak{L}(G)\defined\mathcal{P}(V^G)\times\mathcal{P}(E^G)\times\mathcal{P}(E^G)$. Next, we specify a function $\mathcal{R}_L:\mathfrak{I}\to\mathfrak{R}_L$ which maps a graph-node incidence $(G,u)$ to a ring of size $L$. For this, we need some terminology: The distance $d_V^G(u,u^\prime)$ between two nodes $u,u^\prime\in V^G$ is defined as the number of edges on a shortest path connecting them or as $\infty$ if they are in different connected components of $G$. The eccentricitiy of a node $u\in V^G$ and the diameter of a graph $G$ are defined as $e^G_V(u):=\max_{u^\prime\in V^G}d_V^G(u,u^\prime)$ and $\diam(G):=\max_{u\in V^G}e^G_V(u)$, respectively.

\begin{definition}[Ring]\label{def:ring}
Given a constant $L\in\N_{>0}$, the function $\mathcal{R}_L:\mathfrak{I}\to\mathfrak{R}_L$ maps a graph-node incidence $(G,u)$ to the ring $\RINGLS(G,u):=(\LAYER^G_l(u))^{L-1}_{l=0}$ rooted at $u$ in $G$ (\Cref{fig:ring}). For the dummy node $\epsilon$, we define $\RINGLS(G,\epsilon):=((\emptyset,\emptyset,\emptyset)_l)^{L-1}_{l=0}$. For all other nodes $u$, $\LAYER^G_l(u):=(\NODES^G_l(u),\OUTER^G_l(u),\INNER^G_l(u))$ denotes the $l$\textsuperscript{th} layer rooted at $u$ in $G$, where:
\begin{enumerate}
\item $\NODES^G_l(u):=\{u^\prime\in V^G\mid d_V^G(u,u^\prime)=l\}$ is the set of nodes at distance $l$ from $u$. 
\item $\INNER^G_l(u):=E^G\cap\left(\NODES^G_l(u)\times \NODES^G_l(u)\right)$
is the set of inner edges connecting two nodes in the $l$\textsuperscript{th} layer. 
\item $\OUTER^G_l(u):=E^G\cap\left(\NODES^G_l(u)\times \NODES^G_{l+1}(u)\right)$
is the set of outer edges connecting a node in the $l$\textsuperscript{th} layer to a node in the $(l+1)$\textsuperscript{th} layer.
\end{enumerate}
\end{definition}

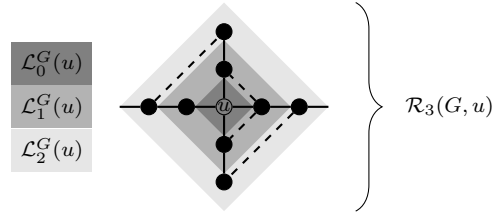
\begin{figure}[t]
\centering
\input{img/rings}
\caption{Visualization of \Cref{def:ring}. Inner edges are dashed, outer edges are solid. Layers are displayed in different shades of grey.}\label{fig:ring}
\end{figure}

It is easy to see that the ring $\RINGLS[1](G,u)$ of a node $u\in V^G$
corresponds to the branch structures used by \BP, \BRANCH, \BRANCHFAST, and \BRANCHUNI. Further
properties of rings and layers are summarized in \Cref{rem:rings}.

\begin{remark}[Properties of Rings]\label{rem:rings}
  Let $u\in V^G$ be a node and
  $\RINGLS(G,u)=((\NODES^G_l(u),\OUTER^G_l(u),\INNER^G_l(u))_l)^L_{l=0}$
  be the ring of size $L$ rooted at $u$.  Then the following
  statements follow from the involved definitions:
\begin{enumerate}
\item The node set $\NODES^G_l(u)$ is empty if and only if $l>e^G_V(u)$, the edge set $\INNER^G_l(u)$ is empty if $l>e^G_V(u)$, and the edge set $\OUTER^G_l(u)$ is empty if and only if $l>e^G_V(u)-1$.
\item All node sets $\NODES^G_l(u)$ and all edge sets $\OUTER^G_l(u)$ and $\INNER^G_l(u)$ are disjoint.
\item The equalities $\bigcup^{L-1}_{l=0}\NODES^G_l(u)=V^G$ and $\bigcup^{L-1}_{l=0}(\OUTER^G_l(u)\cup\INNER^G_l(u))=E^G$ hold for all $u\in V^G$ if and only if $L>\diam(G)$.
\end{enumerate}
\end{remark}

\Cref{prop:poset-soa-rings} and \Cref{fig:poset-soa-rings} show how rings relate to existing local structures in terms of topological information content. We see that rings of size at least $2$ are strictly topologically more informative than the local structures employed by the baseline instantiations \BRANCHUNI, \BRANCHFAST, \BRANCH, \BP, and \STAR. Moreover, for fixed sizes, rings are incomparable to the bags of walks used by \WALKS and strictly topologically less informative than the rooted subgraphs used by \SUBGRAPH. Recall, however, that rooted subgraphs cannot be compared in polynomial time and that bags of walks only model constant edit costs. Rings are hence the only generically applicable and polynomially comparable local structures that provably capture more topological information than the baseline structures. 

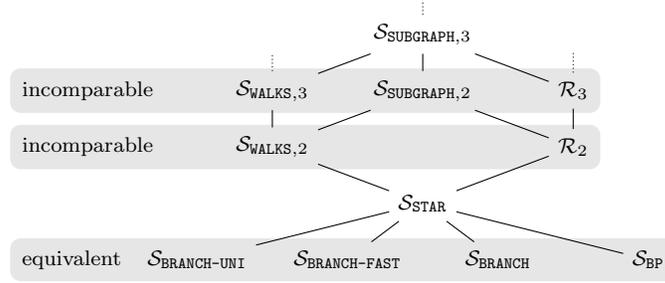
\begin{figure}
\centering
\input{img/poset-soa-rings}
\caption{Hasse diagram induced by strict partial order $\succ_T$ for rings and local structure functions employed by existing instantiations of \LSAPEGED.}\label{fig:poset-soa-rings}
\end{figure}

\begin{proposition}\label{prop:poset-soa-rings}
The strict partial order $\succ_T$ orders the ring functions $\mathcal{R}_L$ and local structure functions employed by existing classical \LSAPEGED instantiations as shown in the Hasse diagram in \Cref{fig:poset-soa-rings}.
\end{proposition}

\begin{proof}
Since rings of size $2$ are stars plus inner and outer edges in the layer with index $1$, we have $\RINGLS[2]\succ_T\mathcal{S}_\STAR$. $\mathcal{S}_{\WALKS,L}\nsucceq_T\RINGLS$ is implied by the fact that bags of walks are oblivious as to how far the nodes on the walks are away from the root. Conversely, rings do not model how nodes in different layers are connected to each others, which implies $\RINGLS\nsucceq_T\mathcal{S}_{\WALKS,L}$. $\RINGLS[L+1]\succ_T\RINGLS$, $\mathcal{S}_{\SUBGRAPH,L}\succeq_T\RINGLS$, and the fact that $\RINGLS[L+1]$ and $\mathcal{S}_{\SUBGRAPH,L}$ are incomparable directly follow from the involved definitions.
\end{proof}

\subsection{Rings: construction}\label{sec:rings:construction}

\Cref{fig:bfs} shows how to construct a ring $\RINGLS(G,u)$ via
breadth-first search. The algorithm maintains the level $l$ of the
currently processed layer along with the layer's node and edge sets
$\NODES$, $\OUTER$, and $\INNER$, a vector $\mathtt{d}$ that stores
for each node $u^\prime\in V^G$ the distance to the root $u$, flags
$\mathtt{discovered}[e]$ that indicate if the edge $e\in E^G$ has
already been discovered by the algorithm, and a FIFO queue
$\mathtt{open}$ which is initialized with the root $u$. Throughout the
algorithm, $\mathtt{d}[u^\prime]=d^G_V(u,u^\prime)$ holds for all
nodes $u^\prime$ which have already been added to $\mathtt{open}$,
while newly discovered nodes $u^{\prime\prime}$ have
$\mathtt{d}[u^{\prime\prime}]=\infty$.

\begin{figure}[t]
\centering
\removelatexerror
\input{img/bfs}
\caption{Construction of rings via breadth-first search.}\label{fig:bfs}
\end{figure}

If a node $u^\prime$ is popped from $\mathtt{open}$, we check if its
distance is larger than the level $l$ of the current layer. If this is
case, we store the current layer, increment $l$, and clear the node
and edge sets $\NODES$, $\OUTER$, and $\INNER$. Next, we add the node
$u^\prime$ to $\NODES$ and iterate through its
undiscovered incident edges $(u^\prime,u^{\prime\prime})$. We mark
$(u^\prime,u^{\prime\prime})$ as discovered and push the node
$u^{\prime\prime}$ to $\mathtt{open}$ if it has not been discovered
yet and its distance to $u$ is less than $L$. If this
distance equals $l$, 
$(u^\prime,u^{\prime\prime})$ is added to the inner edges $\INNER$;
otherwise, it is added to the outer edges $\OUTER$. Once
$\mathtt{open}$ is empty, the last layer is stored and the complete
ring is returned. Since nodes and edges are processed at most once,
the algorithm runs in $O(|V^G|+|E^G|)$ time. The runtime complexity is hence independent of the parameter $L$.
How $L$ affects ring based instantiations of \LSAPEGED and how it should be chosen is detailed in \Cref{sec:heuristics:parameters}.

%% file: img/rings.tex
\tikzset{n/.style={circle,minimum width=6pt,inner sep=0pt}}
\tikzset{e/.style={thick}}
\begin{tikzpicture}[every node/.style={font=\footnotesize}, scale=1, transform shape]
\node[n,draw=none] (u3-1) at (1.5,0) {};
\node[n,draw=none] (u3-2) at (-1.5,0) {};
\node[n,draw=none] (u3-3) at (0,1.5) {};
\node[n,draw=none] (u3-4) at (0,-1.5) {};
\node[n,draw] (u) at (0,0) {$u$};
\node[n,draw,fill=black] (u1-1) at (.5,0) {};
\node[n,draw,fill=black] (u1-2) at (-.5,0) {};
\node[n,draw,fill=black] (u1-3) at (0,.5) {};
\node[n,draw,fill=black] (u1-4) at (0,-.5) {};
\node[n,draw,fill=black] (u2-1) at (1,0) {};
\node[n,draw,fill=black] (u2-2) at (-1,0) {};
\node[n,draw,fill=black] (u2-3) at (0,1) {};
\node[n,draw,fill=black] (u2-4) at (0,-1) {};
\foreach \i in {1,...,4} {
\draw[e] (u) -- (u1-\i);
\draw[e] (u1-\i) -- (u2-\i);
}
\draw[e] (u2-1) -- (u3-1);
\draw[e] (u2-2) -- (u3-2);
\draw[e,dashed] (u1-1) -- (u1-4);
\draw[e,dashed] (u1-1) -- (u1-3);
\draw[e,dashed] (u2-2) -- (u2-3);
\draw[e,dashed] (u2-1) -- (u2-4);

\begin{scope}[on background layer]
\filldraw[fill=black!10,draw=none] (u3-1.west) -- (u3-3.south) -- (u3-2.east) -- (u3-4.north) -- (u3-1.west);
\filldraw[fill=black!30,draw=none] (u2-1.west) -- (u2-3.south) -- (u2-2.east) -- (u2-4.north) -- (u2-1.west);
\filldraw[fill=black!50,draw=none] (u1-1.west) -- (u1-3.south) -- (u1-2.east) -- (u1-4.north) -- (u1-1.west);
\end{scope}
\node[fill=black!30,anchor=east] (L1) at (-1.75,0) {$\LAYER^G_1(u)$};
\node[fill=black!10,anchor=north east] (L2) at (L1.south east) {$\LAYER^G_2(u)$};
\node[fill=black!50,anchor=south east] (L0) at (L1.north east) {$\LAYER^G_0(u)$};
\gettikzxy{(u3-3.south)}{\midx}{\north}
\gettikzxy{(u3-4.north)}{\midx}{\south}
\draw [decorate,decoration={brace,amplitude=10pt},yshift=0pt]
(1.75,\north) -- (1.75,\south) node [black,midway,xshift=1.25cm] {$\RINGLS[3](G,u)$};

\end{tikzpicture}

%% file: img/poset-soa-rings.tex
\begin{tikzpicture}[every node/.style={font=\footnotesize}, scale=1, transform shape]
\node (branchuni) at (-3,0) {$\mathcal{S}_\BRANCHUNI$};
\node (branchfast) at (-1,0) {$\mathcal{S}_\BRANCHFAST$};
\node (branch) at (1,0) {$\mathcal{S}_\BRANCH$};
\node (bp) at (3,0) {$\mathcal{S}_\BP$};
\node (star) at (0,.75) {$\mathcal{S}_\STAR$};
\node (walks2) at (-2,1.5) {$\mathcal{S}_{\WALKS,2}$};
\node (walks3) at (-2,2.25) {$\mathcal{S}_{\WALKS,3}$};
\node (ring2) at (2,1.5) {$\mathcal{R}_{2}$};
\node (ring3) at (2,2.25) {$\mathcal{R}_{3}$};
\node (subgraph2) at (0,2.25) {$\mathcal{S}_{\SUBGRAPH,2}$};
\node (subgraph3) at (0,3) {$\mathcal{S}_{\SUBGRAPH,3}$};

\node[anchor=west] (equiv) at (-5.45,0) {equivalent};
\node[anchor=west] (incomp2) at (-5.45,1.5) {incomparable};
\node[anchor=west] (incomp3) at (-5.45,2.25) {incomparable};
\begin{scope}[on background layer]
\node[fit=(branchuni)(branchfast)(branch)(bp)(equiv),fill=black!10,inner sep=1pt, rounded corners] (box) {};
\node[fit=(ring2)(walks2)(incomp2),fill=black!10,inner sep=1pt, rounded corners] (box) {};
\node[fit=(ring3)(walks3)(incomp3)(subgraph2),fill=black!10,inner sep=1pt, rounded corners] (box) {};
\end{scope}

\draw (branchuni) -- (star);
\draw (branchfast) -- (star);
\draw (branch) -- (star);
\draw (bp) -- (star);
\draw (star) -- (walks2);
\draw (star) -- (ring2);
\draw (ring2) -- (subgraph2);
\draw (ring3) -- (subgraph3);
\draw (ring2) -- (ring3);
\draw (walks2) -- (subgraph2);
\draw (walks2) -- (walks3);
\draw (walks3) -- (subgraph3);
\draw (subgraph2) -- (subgraph3);
\draw[densely dotted] (subgraph3) edge ($(subgraph3)+(0,.5)$);
\draw[densely dotted] (walks3) edge ($(walks3)+(0,.5)$);
\draw[densely dotted] (ring3) edge ($(ring3)+(0,.5)$);
\end{tikzpicture}

%% file: img/bfs.tex
\begin{algorithm}[H]
\SetCommentSty{emph}
\SetKw{Continue}{continue}
\KwIn{Graph $G$, node $u\in V^G$, constant $L\in\N_{>0}$.}
\KwOut{Ring $\RINGLS(G,u)$ rooted at $u$.}
\BlankLine
$l\gets0$; $\NODES\gets\emptyset$; $\OUTER\gets\emptyset$; $\INNER\gets\emptyset$\; 
$\RINGLS(G,u)\gets((\emptyset,\emptyset,\emptyset)_l)^{L-1}_{l=0}$; $\mathtt{open}\gets\{u\}$\;
$\mathtt{d}[u]\gets0$; \lFor{$u^\prime\in V^G\setminus\{u\}$}{$\mathtt{d}[u^\prime]\gets\infty$}
\lFor{$e\in E^G$}{$\mathtt{discovered}[e]\gets\mathtt{false}$}
\While{$\mathtt{open}\neq\emptyset$} {
$u^\prime\gets\mathtt{open.pop}()$\;
\If{$\mathtt{d}[u^\prime]>l$}{
$\RINGLS(G,u)_l\gets(\NODES,\OUTER,\INNER)$; $l\gets l+1$\; 
$\NODES\gets\emptyset$; $\OUTER\gets\emptyset$; $\INNER\gets\emptyset$\;
}
$\NODES\gets\NODES\cup\{u^\prime\}$\;
\For{$(u^\prime,u^{\prime\prime})\in E^G$} {
\lIf{$\mathtt{discovered}[(u^\prime,u^{\prime\prime})]$}{\Continue}
$\mathtt{discovered}[(u^\prime,u^{\prime\prime})]\gets\mathtt{true}$\;
\If{$\mathtt{d}[u^{\prime\prime}]=\infty$}{
$\mathtt{d}[u^{\prime\prime}]\gets l+1$\;
\lIf{$\mathtt{d}[u^{\prime\prime}]<L$}{$\mathtt{open.push}(u^{\prime\prime})$}
}
\lIf{$\mathtt{d}[u^{\prime\prime}]=l$}{$\INNER\gets\INNER\cup\{(u^\prime,u^{\prime\prime})\}$}
\lElse{$\OUTER\gets\OUTER\cup\{(u^\prime,u^{\prime\prime})\}$}
}
}
$\RINGLS(G,u)_l\gets(\NODES,\OUTER,\INNER)$; \Return $\RINGLS(G,u)$\;
\end{algorithm}

%% file: heuristics.tex
\section{Instantiations of \LSAPEGED Based on Rings and Machine Learning}\label{sec:heuristics}

In this section, we present two new heuristics for ring based
transformations to \LSAPE: \RING (\Cref{sec:heuristics:ring}) uses the classical transformation
strategy; \RINGML (\Cref{sec:heuristics:ringml}) employs
our new machine learning based approach. We also discuss how the existing machine learning based approaches \NGM and \PREDICT relate to \LSAPEGED (\Cref{sec:heuristics:ngm-predict}).

\subsection{\RING: a classical instantiation}\label{sec:heuristics:ring}

\RING is a classical instantiation of the paradigm \LSAPEGED which
uses rings of size $L$ as local structures. Therefore, what remains to
be done is to define a distance measure
$d_{\mathfrak{R}_L}:\mathfrak{R}_L\times\mathfrak{R}_L\to\mathbb{R}_{\geq0}$
for the rings. We will define such a distance measure in a bottom-up
fashion: Ring distances are defined in terms of layer distances,
which, in turn, are defined in terms of node and edge set distances.

Assume that, for all pairs of graphs
$(G,H)\in\mathfrak{G}\times\mathfrak{G}$, we have access to measures
$d^{G,H}_{\mathcal{P}(V)}:\mathcal{P}(V^G)\times\mathcal{P}(V^H)\to\mathbb{R}_{\geq0}$
and
$d^{G,H}_{\mathcal{P}(E)}:\mathcal{P}(E^G)\times\mathcal{P}(E^H)\to\mathbb{R}_{\geq0}$
that compute distances between subsets of nodes and edges. Then we can define a layer distance measure
$d^{G,H}_{\mathfrak{L}}:\mathfrak{L}(G)\times\mathfrak{L}(H)\to\mathbb{R}_{\geq0}$
as
\begin{IEEEeqnarray*}{rCl}
d^{G,H}_{\mathfrak{L}}(\mathcal{L}^G,\mathcal{L}^H)&\defined&\frac{\alpha_0d^{G,H}_{\mathcal{P}(V)}(\NODES^G,\NODES^H)}{\max\{|\NODES^G|,|\NODES^H|,1\}}
+\frac{\alpha_1d^{G,H}_{\mathcal{P}(E)}(\INNER^G,\INNER^H)}{\max\{|\INNER^G|,|\INNER^H|,1\}}\\
&&+\>\frac{\alpha_2d^{G,H}_{\mathcal{P}(E)}(\OUTER^G,\OUTER^H)}{\max\{|\OUTER^G|,|\OUTER^H|,1\}}\text{,}
\end{IEEEeqnarray*}
where $\boldsymbol{\alpha}\in\boldsymbol{\Delta}^2$ is a simplex
vector of weights associated to the distances between nodes, inner
edges, and outer edges. We normalize by the sizes of the involved node
and edge sets in order not to overrepresent large layers. Using the
layer distances and a simplex weight vector
$\boldsymbol{\lambda}\in\boldsymbol{\Delta}^{L-1}$, we define the ring distance measure as follows:
\begin{IEEEeqnarray*}{c}
d_{\mathfrak{R}_L}((\LAYER^G_l)^{L-1}_{l=0},(\LAYER^H_l)^{L-1}_{l=0})\defined\sum^{L-1}_{l=0}\lambda_ld^{G,H}_{\mathfrak{L}}(\mathcal{L}^G_l,\mathcal{L}^H_l)
\end{IEEEeqnarray*}

Next, we define node and edge set distances $d^{G,H}_{\mathcal{P}(V)}$ and $d^{G,H}_{\mathcal{P}(E)}$. To obtain tight upper bounds for \GED, they should be defined such that $d_{\mathfrak{R}_L}(\RINGLS(G,u),\RINGLS(H,v))$ is small just in case the node assignment $(G,H,u,v)$ induces a small edit cost. We suggest two strategies that meet this desideratum. 

\subsubsection{\LSAPE based node and edge set distances}\label{sec:heuristics:ring:lsape}

The first approach uses the edit cost functions $c_V$ and $c_E$ to populate \LSAPE instances and then defines the distances in terms of the costs of optimal or greedy \LSAPE solutions. Given node sets $\NODES^G=\{u_1,\ldots,u_{|\NODES^G|}\}\subseteq V^G$ and $\NODES^H=\{v_1,\ldots,v_{|\NODES^H|}\}\subseteq V^H$, an \LSAPE instance $\C\in\mathbb{R}^{(|\NODES^G|+1)\times(|\NODES^H|+1)}$ is defined as 
$c_{i,k} \defined c_V(\ell^G_V(u_i),\ell^H_V(v_k))$, 
$c_{i,|\NODES^H|+1} \defined c_V(\ell^G_V(u_i),\epsilon)$, and
$c_{|\NODES^G|+1,k} \defined c_V(\epsilon,\ell^H_V(v_k))$, 
for all $(i,k)\in[|\NODES^G|]\times[|\NODES^H|]$. Then, a solution $\pi\in\Pi(\C)$ is computed\,---\,either optimally in $O(\min\{|\NODES^G|,|\NODES^H|\}^2\max\{|\NODES^G|,|\NODES^H|\})$ time or greedily in $O(|\NODES^G||\NODES^H|)$ time\,---\,and the node set distance $d^{G,H}_{\mathcal{P}(V)}$ between $\NODES^G$ and $\NODES^H$ is defined as
$d^{G,H}_{\mathcal{P}(V)}(\NODES^G,\NODES^H)\defined\C(\pi)$.
The edge set distance $d^{G,H}_{\mathcal{P}(E)}$ can be defined analogously.

\subsubsection{Multiset based node and edge set distances}\label{sec:heuristics:multiset}
  
Using \LSAPE to define $d^{G,H}_{\mathcal{P}(V)}$ and
$d^{G,H}_{\mathcal{P}(E)}$ yields fine-grained distance measures but
incurs a relatively high computation time. As an alternative, we suggest a faster,
multiset intersection based approach which computes a proxy for the
\LSAPE based distances. For this, the distance between node sets
$\NODES^G\subseteq V^G$ and $\NODES^H\subseteq V^H$ is defined as
\begin{IEEEeqnarray*}{rCl}
d^{G,H}_{\mathcal{P}(V)}&\defined&\delta_{|\NODES^G|>|\NODES^H|}\overline{c}_\mathit{del}(|\NODES^G|-|\NODES^H|)+\delta_{|\NODES^G|<|\NODES^H|}\overline{c}_\mathit{ins}(|\NODES^H|-|\NODES^G|)\\
&&+\>\overline{c}_\mathit{sub}(\min\{|\NODES^G|,|\NODES^H|\}-|\ell^G_V\llbracket
\NODES^G\rrbracket \cap\ell^H_V\llbracket \NODES^G\rrbracket |)\text{,}
\end{IEEEeqnarray*}
where $\overline{c}_\mathit{del}$, $\overline{c}_\mathit{ins}$, and
$\overline{c}_\mathit{sub}$ are the average costs of deleting a node
in $\NODES^G$, inserting a node in $\NODES^H$, and substituting a node
in $\NODES^G$ by a differently labeled node in $\NODES^H$, and
$\ell^G_V\llbracket \NODES^G\rrbracket $ and
$\ell^H_V\llbracket \NODES^H\rrbracket $ are the multiset images of
$\NODES^G$ and $\NODES^H$ under $\ell^G_V$ and
$\ell^H_V$. 

Since multiset intersections can be computed in
quasilinear time \cite{zeng:2009aa}, the dominant operation is the
computation of $\overline{c}_\mathit{sub}$ which requires
$O(|\NODES^G||\NODES^H|)$ time. Again, the edge set distance
$d^{G,H}_{\mathcal{P}(E)}$ can be defined analogously. The following
\Cref{lm:lsape-vs-multiset} relates the \LSAPE based definitions of
$d^{G,H}_{\mathcal{P}(V)}$ and $d^{G,H}_{\mathcal{P}(E)}$ to the ones
based on multiset intersection and justifies our claim that the latter
are proxies for the former.

\begin{proposition}\label{lm:lsape-vs-multiset}
Let $G,H\in\mathfrak{G}$, $\NODES^G\subseteq V^G$, $\NODES^H\subseteq V^H$, and assume that $c_V$ is quasimetric between $\NODES^G$ and $\NODES^G$, \ie, that $c_V(\ell^G_V(u),\ell^H_V(v))\leq c_V(\ell^G_V(u),\epsilon)+c_V(\epsilon,\ell^H_V(v))$ holds for all $(u,v)\in\NODES^G\times\NODES^H$. Then the definitions of $d^{G,H}_{\mathcal{P}(V)}(\NODES^G,\NODES^H)$ based on \LSAPE and multiset intersection incur the same number of node insertions, deletions, and substitutions. If, additionally, there are constants $c_\mathit{del},c_\mathit{ins},c_\mathit{sub}\in\mathbb{R}_{\geq0}$ such that the equations $c_V(\ell^G_V(u),\ell^H_V(v))=c_\mathit{sub}$, $c_V(\ell^G_V(u),\epsilon)=c_\mathit{del}$, and $c_V(\epsilon,\ell^H_V(v))=c_\mathit{ins}$ hold for all $(u,v)\in\NODES^G\times\NODES^H$, the two definitions coincide. For the edge set distances $d^{G,H}_{\mathcal{P}(E)}$, analogous statements hold.
\end{proposition}

\begin{proof}
  Assume \mywlog that $|\NODES^G|\leq |\NODES^H|$, let \C be the
  \LSAPE instance for $\NODES^G$ and $\NODES^H$ constructed as shown in \Cref{sec:heuristics:ring:lsape}, and 
  $\pi$ be an optimal solution for \C. Since $c_V$
  is quasimetric, we know from \cite{bougleux:2020aa} that
  $\pi$ does not contain deletions and contains exactly
  $|\NODES^H|-|\NODES^G|$ insertions. This proves the first part of
  the proposition. If we additionally have constant edit costs between
  $\NODES^G$ and $\NODES^H$, $\C(\pi)$ is
  reduced to the cost of $|\NODES^H|-|\NODES^G|$ insertions plus
  $c_\mathit{sub}=\overline{c}_\mathit{sub}$ times the number of
  non-identical substitutions. This last quantity is provided by
  $|\NODES^G|-|l_V^G\llbracket \NODES^G\rrbracket \cap l_V^H\llbracket
  \NODES^G\rrbracket |$.
  We thus have
  $\C(\pi)=\overline{c}_\mathit{ins}(|\NODES^H|-|\NODES^G|)+\overline{c}_\mathit{sub}(|\NODES^G|-|l_V^G\llbracket
  \NODES^G\rrbracket \cap l_V^H\llbracket \NODES^G\rrbracket |)$,
  as required. The proof for $d^{G,H}_{\mathcal{P}(E)}$ is
  analogous.
\end{proof}

\begin{figure}[t]
\centering
\removelatexerror
\input{img/params}
\caption{Algorithm for learning the parameters $L$, $\boldsymbol{\alpha}$, and $\boldsymbol{\lambda}$.}\label{fig:params}
\end{figure}

\subsubsection{Choice of parameters and runtime complexity}\label{sec:heuristics:parameters}

In \Cref{fig:params}, an algorithm is described that, given a set of training graphs $\mathcal{G}$
and node and edge set distances $d^{G,H}_{\mathcal{P}(V)}$ and
$d^{G,H}_{\mathcal{P}(E)}$, learns good values for $L$,
$\boldsymbol{\alpha}$, and $\boldsymbol{\lambda}$. First,
$L$ is set to an upper bound for the ring sizes and all rings of size $L$ rooted at the
nodes of the graphs $G\in\mathcal{G}$ are constructed (\cf \Cref{fig:bfs}). Next, $L$ is lowered to
$1$ plus the largest $l<L$ such that there is a graph
$G\in\mathcal{G}$ and a node $u\in V^G$ with
$\RINGLS(G,u)_l\neq(\emptyset,\emptyset,\emptyset)$. By
\Cref{rem:rings}, this $l$ equals the maximal diameter of the graphs
contained in $\mathcal{G}$. Let $\RING_{L,\boldsymbol{\alpha},\boldsymbol{\lambda}}(G,H)$ bet the upper bound for $\GED(G,H)$ returned by \RING if called with parameters $L$, $\boldsymbol{\alpha}$ and $\boldsymbol{\lambda}$, and $\mu\in[0,1]$ be a tuning parameter that should be small if one wants to optimize for tightness and large if one wants to optimize for runtime. Then, a blackbox optimizer
\cite{rios:2013aa} is called to minimize the objective 
$f_{L,\mu}(\boldsymbol{\alpha},\boldsymbol{\lambda})\defined
\left[\mu + (1-\mu)\cdot\frac{|\supp(\boldsymbol{\lambda})|-1}{\max\{1,L-1\}}\right]\cdot\sum_{(G,H)\in\mathcal{G}^2}\RING_{\boldsymbol{\alpha},\boldsymbol{\lambda},L}(G,H)$
over all simplex vectors $\boldsymbol{\alpha}\in\boldsymbol{\Delta}^{2}$ and $\boldsymbol{\lambda}\in\boldsymbol{\Delta}^{L-1}$. We include $|\supp(\boldsymbol{\lambda})|-1$ in the objective, because only levels which are contained in the support of $\boldsymbol{\lambda}$ (\ie, all levels $l$ with $\lambda_l>0$) contribute to $d_{\mathfrak{R}_L}$. Hence, only few layer distances have to be computed if $|\supp(\boldsymbol{\lambda})|$ is small. Once optimized parameters $\boldsymbol{\alpha}$ and $\boldsymbol{\lambda}$ have been computed, $L$ can be further lowered to $L=1+\max{\supp(\boldsymbol{\lambda})}$.

\begin{remark}[Runtime Complexity of \RING]\label{prop:ring:runtime}
Let $G,H\in\mathfrak{G}$ be graphs, $L\in\N_{\geq0}$ be a constant, and $\Omega=\mathcal{O}(\min\{\max\{|E^G|,|E^H|\},\max\{\Delta(G),\Delta(H)\}^L\})$ be the size of the largest node or edge set contained in the rings of $G$ and $H$, where $\Delta(G)$ is $G$'s maximum degree. Upon constructing the rings, \RING requires $\mathcal{O}(\Omega^3|V^G||V^H|)$ time to populate its \LSAPE instance, if \LSAPE based node and edge set distances are used, and $\mathcal{O}(\Omega^2|V^G||V^H|)$ time, if multiset based distances are employed.
\end{remark}


\subsection{\RINGML: a machine learning based instantiation}\label{sec:heuristics:ringml}

If \LSAPEGED is instantiated with the help of machine learning
techniques, feature vectors associated to the node assignments have to
be defined. The heuristic \RINGML uses rings of size $L$ to accomplish
this task. Formally, \RINGML defines a function
$\mathcal{F}:\mathfrak{A}\to\mathbb{R}^{6L+10}$ that maps node
assignments to feature vectors with six features per layer and ten
global features. Let $(G,H,u,v)\in\mathfrak{A}$ be a node assignment
and $\RINGLS(G,u)$ and $\RINGLS(H,v)$ be the rings rooted at $u$ in
$G$ and at $v$ in $H$, respectively. For each level
$l\in\{0,\ldots,L{-}1\}$, a feature vector
$\mathbf{x}^l\in\mathbb{R}^6$ is constructed by comparing the layers
$\RINGLS(G,u)_l=(\NODES^G_l(u),\OUTER^G_l(u),\INNER^G_l(u))$ and
$\RINGLS(H,v)_l=(\NODES^H_l(v),\OUTER^H_l(v),\INNER^H_l(v))$ at level
$l$:
$\mathbf{x}^l_0 \defined |\NODES^G_l(u)|-|\NODES^H_l(v)|$,
$\mathbf{x}^l_1 \defined |\OUTER^G_l(u)|-|\OUTER^H_l(v)|$,
$\mathbf{x}^l_2 \defined |\INNER^G_l(u)|-|\INNER^H_l(v)|$,
$\mathbf{x}^l_3 \defined d^{G,H}_{\mathcal{P}(V)}(\NODES^G_l(u),\NODES^H_l(v))$,
$\mathbf{x}^l_4 \defined d^{G,H}_{\mathcal{P}(E)}(\OUTER^G_l(u),\OUTER^H_l(v))$, 
$\mathbf{x}^l_5 \defined d^{G,H}_{\mathcal{P}(E)}(\INNER^G_l(u),\INNER^H_l(v))$.

The first three features compare the layers' topologies. The last
three features use node and edge set distances
$d^{G,H}_{\mathcal{P}(V)}$ and $d^{G,H}_{\mathcal{P}(E)}$ to express the similarity of the involved node
and edge labels. \RINGML also constructs a vector
$\mathbf{x}^{G,H}\in\mathbb{R}^{10}$ of ten global features: the number of nodes and edges of $G$ and $H$,
the average costs for deleting nodes and edges from $G$, the average
costs for inserting nodes and edges into $H$, and the average costs
for substituting nodes and edges in $G$ by nodes and edges in $H$. The
complete feature vector $\mathcal{F}(G,H,u,v)$ is then defined as the
concatenation of the global features $\mathbf{x}^{G,H}$ and the layer
features $\mathbf{x}^l$.

\begin{remark}[Runtime Complexity of \RINGML]\label{prop:ring-ml:runtime}
Let $G,H\in\mathfrak{G}$ be graphs, $L\in\N_{\geq0}$ be a constant, and $\Omega$ be the size of the largest node or edge set contained in one of the rings of $G$ and $H$. Then, once all rings have been constructed, \RINGML requires $\mathcal{O}((\Omega^3+f^{\mathtt{ML}}(1))|V^G||V^H|)$ time to populate its \LSAPE instance \C if \LSAPE based node and edge set distances are used, and $\mathcal{O}((\Omega^2+f^{p^\star}(1))|V^G||V^H|)$ time if multiset intersection based distances are employed. $\mathcal{O}(f^{p^\star}(n))$ is the complexity of evaluating the probability estimate $p^\star$ of the chosen machine learning technique on feature vectors of size $n$.
\end{remark}


\subsection{\NGM and \PREDICT in the context of \LSAPEGED}\label{sec:heuristics:ngm-predict}

To conclude this section, we summarize the existing machine learning based heuristics
 \NGM and \PREDICT, and discuss them in the context of the paradigm \LSAPEGED. 
 
\NGM assumes the node labels to be real-valued vectors and 
defines the feature vectors $\mathcal{F}(G,H,u,v)\defined(\ell^G_V(u),\deg^G(u),\ell^H_V(v),\deg^H(v))$. Hence, no feature
vectors for node deletions and insertions can be constructed. Next, \NGM trains a \DNN to obtain probability estimates for a 
node assignment being good as described in \Cref{sec:paradigm:ml}, and uses these estimates to populate an \LSAP (not \LSAPE) instance. This instance is solved 
to obtain an upper bound for $\GED(G,H)$ in the case where $|V^G|=|V^H|$. \NGM cannot be generalized to the case $|V^G|\neq|V^H|$,
because the last row and column of the \LSAPE instance cannot be populated.

Unlike \NGM, \PREDICT defines feature vectors that cover
node deletions and insertions and are defined for general node and
edge labels. \PREDICT first calls \BP to construct an \LSAPE instance \C. 
Next, $\mathcal{F}(G,H,u,v)$ is defined as the concatenation of four global statistics of \C, the node and the edge costs encoded in the cell $c_{i,k}$ 
that corresponds to the node assignment $(G,H,u,v)$, and ten statistics of the $i$\textsuperscript{th} row and the $k$\textsuperscript{th} column of \C. \PREDICT then trains a 
kernelized \SVC without probability estimates to learn a decision function, which is used to predict if a node assignment is $\epsilon$-optimal. \PREDICT can hence easily be extended to fully instantiate the paradigm \LSAPEGED: 
we only have to replace the \SVC by a \DNN or a \SVM as detailed in \Cref{sec:paradigm:ml}.

%% file: img/params.tex
\begin{algorithm}[H]
\SetCommentSty{emph}
\SetKw{Continue}{continue}
\KwIn{Set of graphs $\mathcal{G}$, node and edge set distances $d^{G,H}_{\mathcal{P}(V)}$ and $d^{G,H}_{\mathcal{P}(E)}$, tuning parameter $\mu$.}
\KwOut{Optimized parameters $L$, $\boldsymbol{\alpha}$, $\boldsymbol{\lambda}$.}
\BlankLine
$L\gets1+\max_{G\in\mathcal{G}}|V^G|$\;
build rings $\RINGLS(G,u)$ for all $G\in\mathcal{G}$ and all $u\in V^G$\;
$L\gets1+\max_{G\in\mathcal{G}}\diam(G)$\tcp*{can be computed in step 2}
$(\boldsymbol{\alpha},\boldsymbol{\lambda})\gets\argmin\{f_{L,\mu}(\boldsymbol{\alpha},\boldsymbol{\lambda})\mid\boldsymbol{\alpha}\in\boldsymbol{\Delta}^{2},\boldsymbol{\lambda}\in\boldsymbol{\Delta}^{L-1}\}$\;
$L\gets1+\max{\supp(\boldsymbol{\lambda})}$\;
\end{algorithm}

%% file: experiments.tex
\section{Experimental Evaluation}\label{sec:exp}

We carried out extensive experiments to empirically evaluate the newly proposed algorithms. In \Cref{sec:exp:setup}, we describe the experimental setup; in \Cref{sec:exp:results}, we report the results.

\subsection{Setup of the experiments}\label{sec:exp:setup}

\subsubsection{Compared methods}\label{sec:exp:methods}

We tested three variants of \RING: \RINGOPT uses optimal \LSAPE
for defining the set distances $d^{G,H}_{\mathcal{P}(V)}$ and
$d^{G,H}_{\mathcal{P}(E)}$, \RINGGD uses greedy \LSAPE, and \RINGMS
uses the multiset intersection based approach. \RINGML was
tested with three different machine learning techniques: \SVC{s} with
RBF kernel and probability estimates \cite{riesen:2016aa}, fully
connected feedforward \DNN{s} \cite{cortes:2018aa}, and \SVM{s} with
RBF kernel. We compared to \LSAPE based competitors that can cope with non-uniform edit costs: 
\BP, \BRANCH, \BRANCHFAST, \SUBGRAPH, \WALKS, and \PREDICT. As \WALKS assumes
constant edit costs, we slightly
extended it by averaging the costs before each run.  To handle \SUBGRAPH's exponential complexity, we set a time limit of \SI{1}{\milli\second}
for computing a cell $c_{i,k}$ of its \LSAPE instance \C. \PREDICT was tested with the same probability estimates as \RINGML. Since some of
our test graphs have symbolic labels and not all of them are of the
same size, we did non include \NGM. For all methods, we
varied the number of threads and \LSAPE
solutions over $\{1,4,7,10\}$ and parallelized the
construction of \C. Moreover, we included the local search based algorithm \IPFP in the experiments, which is one of the most accurate available \GED heuristics but is much slower than instantiations of \LSAPEGED, as detailed in \Cref{sec:intro:related}.

\subsubsection{Benchmark datasets}\label{sec:exp:datasets}

We tested on the benchmark datasets \pah, \alkane, \letter, and \aids
\cite{riesen:2008aa}, which are widely used in the research community \cite{riesen:2009aa,serratosa:2015aa,carletti:2015aa,riesen:2015aa,bougleux:2017aa,daller:2018aa,zheng:2015aa,blumenthal:2020ac,blumenthal:2018aa,blumenthal:2018ac}. \Cref{tab:datasets} summarizes some of their properties. \letter contains graphs that model highly
distorted letter drawings, while the graphs in \pah, \alkane, and \aids
represent chemical compounds. For \letter, we used the edit costs
suggested in Ref.~\refcite{riesen:2010aa}, for \pah, \alkane, and \aids the edit costs
defined in Ref.~\refcite{abu-aisheh:2017aa}. The graphs contained in \pah and \alkane have unlabeled nodes, \ie, \pah and \alkane contain graphs whose information is exclusively encoded in the topologies. \letter and \aids graphs have a higher node informativeness.

\begin{table}[t]
\centering
\caption{Properties of benchmark datasets.}\label{tab:datasets}
\small
\input{tables/datasets}
\end{table}

\subsubsection{Synthetic datasets}

We also tested on synthetic datasets to evaluate the effect of the node informativeness in a controlled setting. For this, we generated datasets \smol{$|\Sigma_V|$}, for $|\Sigma_V|\in\{1,4,7,10\}$, and datasets \sacyclic{$|\Sigma_V|$}, for $|\Sigma_V|\in\{3,5,7,9\}$. The \smol{$|\Sigma_V|$} datasets contain synthetic molecules similar to the ones contained in \alkane. We generated the molecules as pairwise non-isomorphic trees whose sizes were randomly drawn from $\{8,9,10,11,12\}$. Next, we generated four variants of each molecule\,---\,one for each of the four datasets \smol{1}, \smol{4}, \smol{7}, and \smol{10}\,---\,by randomly drawing node labels from $\Sigma_V=[1]$, $\Sigma_V=[4]$, $\Sigma_V=[7]$, and $\Sigma_V=[10]$, respectively. Edges are unlabeled in all variants. The \sacyclic{$|\Sigma_V|$} datasets were generated similarly. More precisely, \sacyclic{3} contains the real-world molecular graphs from the dataset \acyclic, a widely used benchmark dataset with $|\Sigma_V|=3$ and $|\Sigma_E|=2$. The datasets \sacyclic{5}, \sacyclic{7}, and \sacyclic{9} contain variants of the molecules, where the node labels were randomly drawn from $\Sigma_V=[5]$, $\Sigma_V=[7]$, and $\Sigma_V=[9]$. The synthetic datasets are hence constructed such that node label informativeness increases with increasing $|\Sigma_V|$. Since rings are designed for graphs where most information resides in the topologies, we  expect the tightness gain of ring based heuristics \wrt \BP, \BRANCH, and \BRANCHFAST to drop with increasing $|\Sigma_V|$.

\subsubsection{Meta-parameters and training}\label{sec:exp:param}

For learning the meta-parameters of \RINGOPT, \RINGGD, \RINGMS,
\SUBGRAPH, and \WALKS, and training the \DNN{s}, the \SVC{s}, and the
\SVM{s}, we picked a training set $\mathcal{S}_1\subset\mathcal{D}$
with $|\mathcal{S}_1|=50$ for each dataset
$\mathcal{D}$. Following Ref.~\refcite{carletti:2015aa,gauzere:2014aa}, we
picked the parameter $L$ of \SUBGRAPH and \WALKS by
minimizing the mean upper bound on $\mathcal{S}_1$ over
$L\in\{1,2,3,4,5\}$. For choosing the meta-parameters of the \RING variants, we set the tuning parameter $\mu$ to $1$ and initialized our
blackbox optimizer with 100 random simplex vectors
$\boldsymbol{\alpha}$ and $\boldsymbol{\lambda}$. 
For determining the network structure of
the fully connected feedforward \DNN{s}, we carried out $5$-fold cross
validation, varying the number of hidden layers, the number of neurons
per hidden layers, and the activation function at hidden layers over
the grid
$\{1,\ldots,10\}\times\{1,\ldots,20\}\times\{\text{ReLU},\text{Sigmoid}\}$. Similarly,
we determined the meta-parameters $C$ and $\gamma$ of the \SVC via
$5$-fold cross-validation over
$\{10^{-3},\ldots,10^{3}\}\times\{10^{-3},\ldots,10^{3}\}$. For the
\SVM, we set $\gamma=1/\dim(\mathcal{F})$,
where $\dim(\mathcal{F})$ is the number of features. \IPFP was used to compute $\epsilon$-optimal node maps, with an empirically determined $\epsilon\approx0.0423$ \cite{blumenthal:2020aa}. For balancing
the training data $\mathcal{T}$, we
randomly picked only $|\pi|$ node assignments $(u,v)\notin\pi$ for
each close to optimal node map $\pi$. \IPFP's meta-parameters were set to $\kappa=40$, $L=3$, and $\rho=1/4$ (\cf \cite{boria:2020aa} for explanations of $\kappa$, $L$, and $\rho$).

\subsubsection{Protocol, test metrics, and implementation}\label{sec:exp:protocol}

For each dataset $\mathcal{D}$, we randomly selected a test set $\mathcal{S}_2\subseteq\mathcal{D}\setminus\mathcal{S}_1$ with $|\mathcal{S}_2|=\min\{100,|\mathcal{D}\setminus\mathcal{S}_1|\}$, and ran each method on each pair $(G,H)\in\mathcal{S}_2\times\mathcal{S}_2$ with $G\neq H$. We recorded the average runtime in seconds ($t$), the average value of the returned upper bound for \GED ($b$), and the ratio of graphs which are correctly classified if the returned upper bound is employed in combination with a $1$-NN classifier ($r$). For the experiments on the synthetic datasets, we also recorded the deviation in percent from the upper bound computed by the accurate but slow local search algorithm \IPFP ($d$). Having access to a heuristic \ALG that yields tight upper bounds is important, because the tighter the upper bound, the higher the recall if \ALG is used to approximately answer queries of the form \enquote{given a query graph $H$, a graph collection $\mathcal{G}$, and a threshold $\tau$, find all graphs $G\in\mathcal{G}$ with $\GED(G,H)\leq\tau$}.

All methods were implemented in C++ \cite{blumenthal:2019aa}. We employed the \LSAPE solver \cite{bougleux:2020aa}, used NOMAD \cite{le-digabel:2011aa} as our blackbox optimizer, LIBSVM \cite{chang:2011aa} for implementing \SVC{s} and \SVM{s}, and FANN \cite{nissen:2003aa} for implementing \DNN{s}. Tests were run on a machine with two Intel Xeon E5-2667 processors with 8 cores and 98 GB of main memory. Sources and datasets are available at \url{https://github.com/dbblumenthal/gedlib/}.

\subsection{Results of the experiments}\label{sec:exp:results}

\subsubsection{Effect of machine learning techniques}\label{sec:exp:results:ml}

\begin{table*}[t]
\centering
\caption{Effect of machine learning techniques on \RINGML and \PREDICT.}\label{tab:prob-est}
\small
\input{tables/prob-est}
\end{table*}

\Cref{tab:prob-est} shows the performances of different machine learning techniques when used in combination with the feature vectors defined by \RINGML and \PREDICT on the datasets \letter and \pah (with number of threads and maximal number of \LSAPE solutions set to 10). Starred techniques are presented in this paper. The results for \alkane and \aids are similar and are omitted because of space constraints. The tightest upper bounds and best classification ratios were achieved by \SVM{s}. Using \DNN{s} improved the runtime, but resulted in dramatically worse classification ratios and upper bounds. Using \SVC{s} instead of \SVM{s} negatively affected all three test metrics. In the following, we therefore only report the results for \SVM{s} and \DNN{s}. 

\subsubsection{Effect of number of threads and \LSAPE solutions}\label{sec:exp:results:sol-threads}

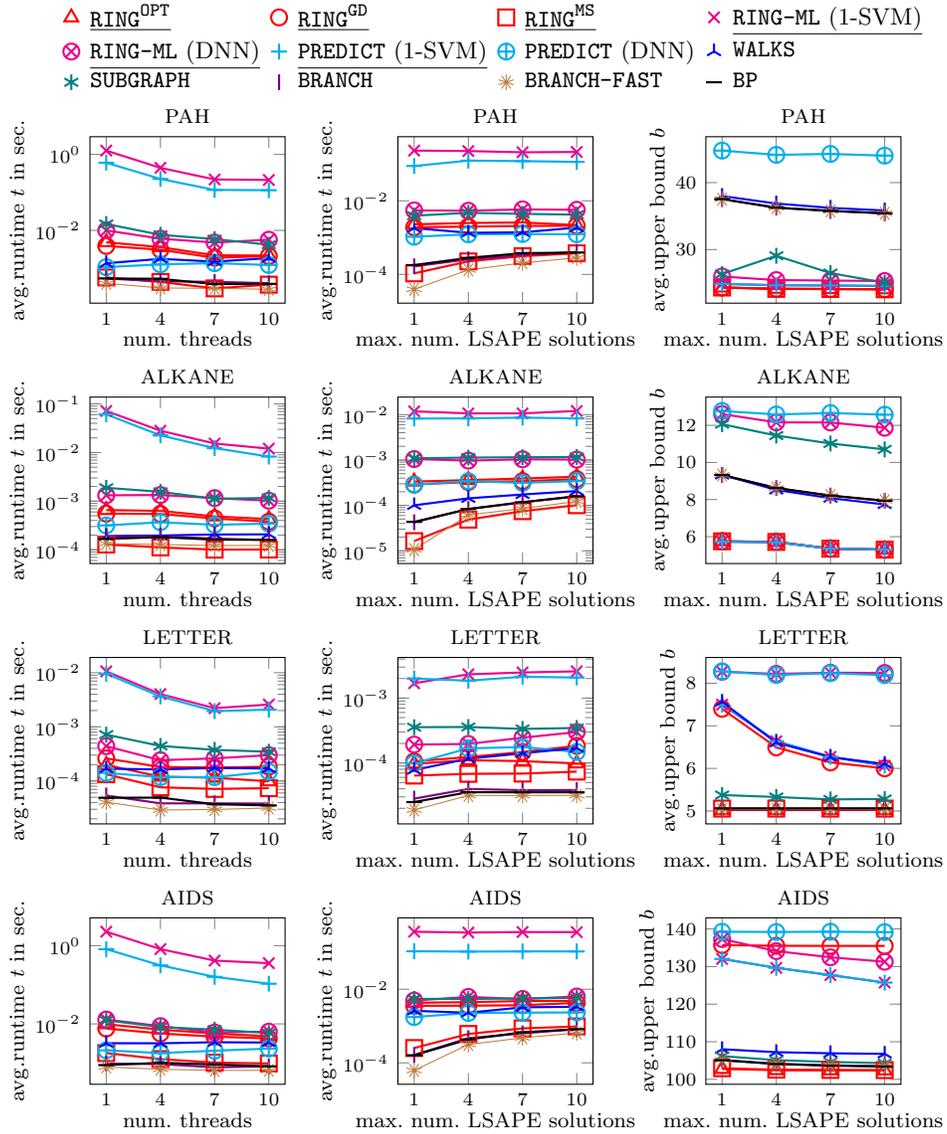
\begin{figure*}[!t]
\centering
\input{img/threads_max_sol_vs_runtimes}
\caption{Number of threads vs.\@ runtimes (first row, maximal number of \LSAPE solutions fixed to 10) and maximal number of \LSAPE solutions vs.\@ runtimes and upper bounds (second and third row, number of threads fixed to 10). Underlined methods use techniques proposed in this paper.}\label{fig:lineplots}
\end{figure*}

\Cref{fig:lineplots} shows the effects of varying the number of threads and \LSAPE solutions.  Unsurprisingly, slower methods benefited more from
parallelization than faster ones. The only exception is \WALKS, whose local structure
distances require a lot of unparallelizable pre-computing. Computing several \LSAPE solution tightened
the upper bounds of mainly those methods that yielded loose upper
bounds if run with only one solution. The outlier of \SUBGRAPH on
\letter is due to the fact that \SUBGRAPH was run with a time
limit on the computation of the subgraph distances; and that the
employed subproblem solver behaves deterministically only if run to optimality. Increasing the number of \LSAPE solutions significantly
increased the runtimes of only the fastest algorithms. Computationally expensive \LSAPEGED instantiations spend most of the runtime on constructing the \LSAPE instances. For these methods, the additional time required to enumerate the \LSAPE solutions is negligible.

\subsubsection{Overall performance on benchmark datasets}\label{sec:exp:results:benchmark}

\begin{figure*}[t]
\centering
\input{img/paretoplots}
\caption{Runtimes vs.\@ upper bounds and classification ratios with number of threads and maximal number of \LSAPE solutions fixed to $10$. Underlined methods use techniques proposed in this paper. As \alkane contains only one graph class, on this dataset, we have $r(\ALG)=1$ for all algorithms \ALG.}\label{fig:paretoplots}
\end{figure*}
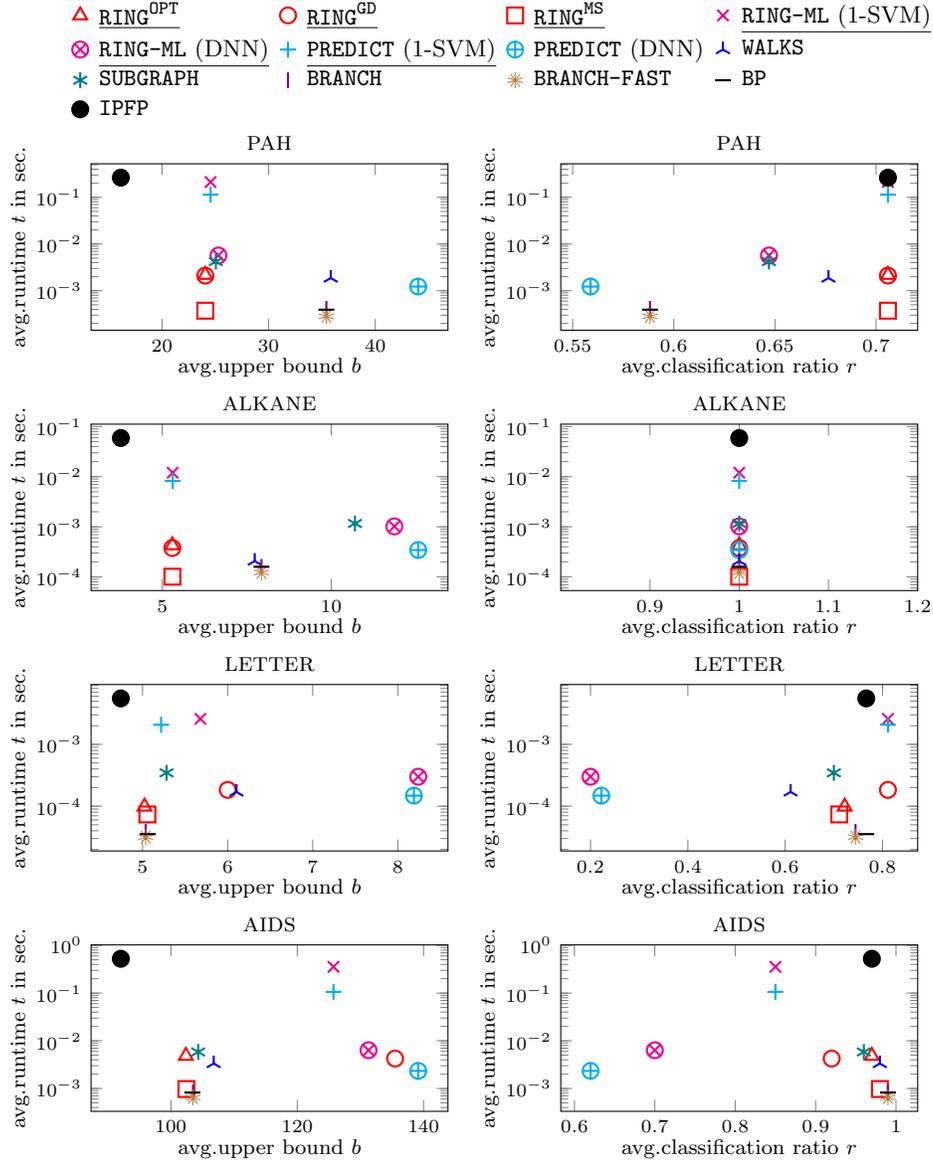

\Cref{fig:paretoplots} summarizes the overall performances of the compared methods on the four benchmark datasets with the number of threads and maximal number of \LSAPE solutions fixed to 10. We see that, across all datasets, \RINGOPT yielded the tightest upper bound among all instantiations of \LSAPEGED. \RINGMS, \ie, the variant of \RING which uses the multiset intersection based approach for computing the layer distances, performed excellently, too, as it was significantly faster than \RINGOPT and yielded only slightly looser upper bounds. The variant \RINGGD performed worse than \RINGOPT and \RINGMS. As expected, the local search algorithm \IPFP computed the tightest upper bounds but was around two orders of magnitude slower than the variants of \RING. 

While \pah and \alkane contain unlabeled graphs whose information is entirely encoded in the topology, \letter and \aids graphs have highly discriminative node labels (\cf column $|\Sigma_V|$ in \Cref{tab:datasets}). This leads to different results for \pah and \alkane, on the one hand, and \letter and \aids, on the other hand. On \letter and \aids, also the baseline approaches \BP, \BRANCH, and \BRANCHFAST yield upper bounds which are only slightly looser than \IPFP's close to optimal bounds. Consequently, the tightness gains of the \RING variants are marginal. For \pah and \alkane, the outcomes are very different. Here, the baseline approaches computed much looser upper bounds than \IPFP and were clearly outperformed by the \RING variants. Using rings hence indeed significantly improves \LSAPE based heuristics on instances which are difficult to solve because the graphs' topologies are more important than the node labels.

If run with \SVM{s} with RBF kernels, the machine learning based methods \PREDICT and \RINGML performed very similarly in terms of classification ratio and tightness of the produced upper bounds. Both yielded very promising classification ratios on \pah and \letter, but were not competitive \wrt runtime, because, at runtime, $\lVert\mathbf{x}^i-\mathcal{F}(G,H,u,v)\rVert^2_2$ has to be computed for each node assignment $(G,H,u,v)$ and each training vector $\mathbf{x}^i$ (\cf Section~\ref{sec:paradigm:ml:technique}). Running \RINGML and \PREDICT with \DNN{s} instead of \SVM{s} dramatically improved the runtimes but led to looser upper bounds and worse classification ratios. If run with \DNN{s}, \RINGML produced tighter upper bounds than \PREDICT. Yet, globally, the machine learning based methods were outperformed by classical \LSAPEGED instantiations.

\subsubsection{Results for synthetic datasets}\label{sec:exp:results:s-mol}

\Cref{fig:smolplots} shows the results for the synthetic datasets. The curves visualize by how much the upper bounds computed by \LSAPE based heuristics deviate from \IPFP's close to optimal upper bounds. Recall that small $|\Sigma_V|$ means that most information resides in the graphs' topologies. To improve the readability, we only show the curves for \RINGOPT, \RINGMS, and the methods \BP, \BRANCH, and \BRANCHFAST that employ narrow local structures. All other methods yielded looser upper bounds or were much slower. 

\RINGOPT and \RINGMS performed very similarly and always yielded the best upper bounds.
As expected, the tightness gap between \RINGOPT and \RINGMS, on the one side, and \BP, \BRANCH, and \BRANCHFAST, on the other side, was much higher on the datasets \smol{1} and \sacyclic{3} than on the datasets \smol{$|\Sigma_V|$} and \sacyclic{$|\Sigma_V|$} for $|\Sigma_V|>3$. The tests on synthetic graphs hence confirms our explanation for the observed results on the benchmark datasets: Using ring based heuristics is especially beneficial if most information is encoded in the graphs' topologies rather than in the node labels. 


\begin{figure}[t]
\centering
\input{img/smolplots}
\caption{Effect of node informativeness with number of threads and maximal number of \LSAPE solutions fixed to $10$. Underlined methods use techniques proposed in this paper. Methods whose curves are not displayed yielded higher deviations or were much slower.}\label{fig:smolplots}
\end{figure}
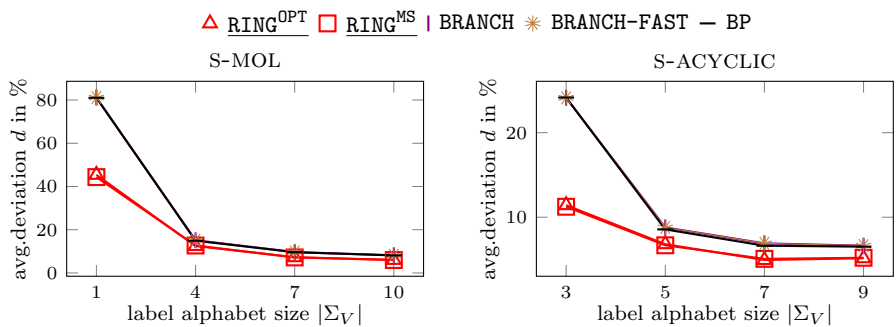

\subsubsection{Upshot of the results}\label{sec:exp:results:upshot}

The experiments yield
four take-home messages: Firstly, \RINGOPT and \RINGMS are the methods of choice if one wants to compute tight upper bounds for \GED on graphs where most information is encoded in the topology but cannot afford to run much slower local search based heuristics such as \IPFP. Secondly, one should always compute several
\LSAPE solutions. This only slightly
increases the runtime and at the same time significantly improves the
upper bounds of methods that yield loose upper bounds if run with only
one \LSAPE solution. Thirdly, machine learning based \LSAPEGED
instantiations such as \RINGML and \PREDICT should be run with \SVM{s}
as suggested in this paper if one wants to optimize for classification
ratio and tightness of the produced upper bound, and with \DNN{s} as
suggested in Ref.~\refcite{cortes:2018aa} if one wants to optimize for runtime. Fourthly, \RINGML and \PREDICT show promising potential but
cannot yet compete with classical instantiations of \LSAPEGED. If run
with \SVM{s}, they are competitive in terms of classification
ratio and tightness of the produced upper bound but not in terms of
runtime; if run with \DNN{s}, the opposite is the case. The open
challenge for future work is therefore to develop new machine learning
frameworks that exploit the information encoded in \RINGML's and
\PREDICT's feature vectors such that the resulting \GED heuristics are
competitive both \wrt quality and \wrt runtime behaviour.

%% file: tables/datasets.tex
\begin{tabular}{@{}lS[table-format = 2.1]S[table-format = 2.0]S[table-format = 2.1]S[table-format = 3.0]S[table-format = 2.0]S[table-format = 2.0]@{}}
\toprule
dataset & {avg. $|V^G|$} & {max. $|V^G|$} & {avg. $|E^G|$} & {max. $|E^G|$} & {$|\Sigma_V|$} & {classes} \\
\midrule
\aids & 15.7 & 95 & 16.2 & 103 & 19 & 2\\
\letter & 4.7 & 9 & 4.5 & 9 & {$\infty$} & 15 \\
\pah & 20.7 & 28 & 24.4 & 34 & 1 & 2 \\
\alkane & 8.9 & 10 & 7.9 & 9 & 1 & 1 \\ 
\bottomrule
\end{tabular}

%% file: tables/prob-est.tex
\begin{tabular}{@{}lS[table-format=2.2]S[table-format=2.2]S[table-format=1.2e-1]S[table-format=1.2e-1]S[table-format=1.2]S[table-format=1.2]@{}}
\toprule
 & {\RINGML{}$^\star$} & {\PREDICT \cite{riesen:2016aa}} & {\RINGML{}$^\star$} & {\PREDICT \cite{riesen:2016aa}} & {\RINGML{}$^\star$} & {\PREDICT \cite{riesen:2016aa}} \\
\midrule
& \multicolumn{6}{c}{\letter}\\
& \multicolumn{2}{c}{avg.\@ upper bound $b$} & \multicolumn{2}{c}{avg.\@ runtime $t$ in sec.} & \multicolumn{2}{c}{avg.\@ classif. ratio $r$} \\
\cmidrule(lr){2-3}\cmidrule(lr){4-5}\cmidrule(l){6-7}
\DNN \cite{cortes:2018aa} & 8.24 & 8.19 & 2.99e-4 & 1.48e-4 & 0.20 & 0.22\\
\SVC \cite{riesen:2016aa} & 6.07 & 6.07 & 6.47e-3 & 2.82e-3 & 0.73  & 0.76\\
\SVM{}$^\star$ & 5.68 & 5.22 & 2.58e-3 & 2.07e-3 & 0.81 & 0.81\\
\midrule
& \multicolumn{6}{c}{\pah}\\
& \multicolumn{2}{c}{avg.\@ upper bound $b$} & \multicolumn{2}{c}{avg.\@ runtime $t$ in sec.} & \multicolumn{2}{c}{avg.\@ classif. ratio $r$} \\
\cmidrule(lr){2-3}\cmidrule(lr){4-5}\cmidrule(l){6-7}
\DNN \cite{cortes:2018aa} & 25.29 & 44.03 & 5.69e-3 & 1.23e-3 & 0.64 & 0.56\\
\SVC \cite{riesen:2016aa} & 31.91 & 36.68 & 7.19e-1 & 3.40e-1 & 0.61 & 0.65\\
\SVM{}$^\star$& 24.55 & 24.55 & 2.12e-1 & 1.14e-1 & 0.71 & 0.71\\
\bottomrule
\end{tabular}

%% file: img/threads_max_sol_vs_runtimes.tex
\begin{tikzpicture}
\begin{groupplot}
[
group style={group name=lineplots, group size=3 by 4, horizontal sep=1.5cm, vertical sep=1.25cm},
width=.33\linewidth,
height=.3\linewidth,
legend columns=4,
legend cell align=left,
legend style={align=left, draw=none, column sep=.5ex, font=\small}
]
\addloglineplots{\pah}{pah__RUNTIMES_vs_THREADS_S10.csv}{num_threads}{avg_runtime}{num.\ threads}{avg.\@ runtime $t$ in sec.}
\addloglineplots{\pah}{pah__RUNTIMES_UB_vs_SOLUTIONS_T10.csv}{num_solutions}{avg_runtime}{max.\ num.\ \LSAPE solutions}{avg.\@ runtime $t$ in sec.}
\addlineplots{\pah}{pah__RUNTIMES_UB_vs_SOLUTIONS_T10.csv}{num_solutions}{avg_ub}{max.\ num.\ \LSAPE solutions}{avg.\@ upper bound $b$}
\addloglineplots{\alkane}{alkane__RUNTIMES_vs_THREADS_S10.csv}{num_threads}{avg_runtime}{num.\ threads}{avg.\@ runtime $t$ in sec.}
\addloglineplots{\alkane}{alkane__RUNTIMES_UB_vs_SOLUTIONS_T10.csv}{num_solutions}{avg_runtime}{max.\ num.\ \LSAPE solutions}{avg.\@ runtime $t$ in sec.}
\addlineplots{\alkane}{alkane__RUNTIMES_UB_vs_SOLUTIONS_T10.csv}{num_solutions}{avg_ub}{max.\ num.\ \LSAPE solutions}{avg.\@ upper bound $b$}
\addloglineplots{\letter}{Letter_HIGH__RUNTIMES_vs_THREADS_S10.csv}{num_threads}{avg_runtime}{num.\ threads}{avg.\@ runtime $t$ in sec.}
\addloglineplots{\letter}{Letter_HIGH__RUNTIMES_UB_vs_SOLUTIONS_T10.csv}{num_solutions}{avg_runtime}{max.\ num.\ \LSAPE solutions}{avg.\@ runtime $t$ in sec.}
\addlineplots{\letter}{Letter_HIGH__RUNTIMES_UB_vs_SOLUTIONS_T10.csv}{num_solutions}{avg_ub}{max.\ num.\ \LSAPE solutions}{avg.\@ upper bound $b$}
\addloglineplots{\aids}{AIDS__RUNTIMES_vs_THREADS_S10.csv}{num_threads}{avg_runtime}{num.\ threads}{avg.\@ runtime $t$ in sec.}
\addloglineplots{\aids}{AIDS__RUNTIMES_UB_vs_SOLUTIONS_T10.csv}{num_solutions}{avg_runtime}{max.\ num.\ \LSAPE solutions}{avg.\@ runtime $t$ in sec.}
\addlineplots{\aids}{AIDS__RUNTIMES_UB_vs_SOLUTIONS_T10.csv}{num_solutions}{avg_ub}{max.\ num.\ \LSAPE solutions}{avg.\@ upper bound $b$}
\end{groupplot}
\node at ($(lineplots c1r1.north) !.5! (lineplots c3r1.north)$) [inner sep=0pt,anchor=south, yshift=3ex] {\pgfplotslegendfromname{grouplegend}};
\end{tikzpicture}

%% file: img/paretoplots.tex
\begin{tikzpicture}
\begin{groupplot}
[
group style={group name=paretoplots, group size=2 by 4, horizontal sep=1.5cm, vertical sep=1.25cm},
width=.5\linewidth,
height=.3\linewidth,
legend columns=4,
legend cell align=left,
legend style={align=left, draw=none, column sep=.5ex, font=\small}
]
\addlogparetoplots{\pah}{pah__RESULTS_T10_S10.csv}{avg_ub}{avg_runtime}{avg.\@ upper bound $b$}{avg.\@ runtime $t$ in sec.}
\addlogparetoplots{\pah}{pah__RESULTS_T10_S10.csv}{avg_classification_ratio}{avg_runtime}{avg.\@ classification ratio $r$}{avg.\@ runtime $t$ in sec.}
\addlogparetoplots{\alkane}{alkane__RESULTS_T10_S10.csv}{avg_ub}{avg_runtime}{avg.\@ upper bound $b$}{avg.\@ runtime $t$ in sec.}
\addlogparetoplots{\alkane}{alkane__RESULTS_T10_S10.csv}{avg_classification_ratio}{avg_runtime}{avg.\@ classification ratio $r$}{avg.\@ runtime $t$ in sec.}
\addlogparetoplots{\letter}{Letter_HIGH__RESULTS_T10_S10.csv}{avg_ub}{avg_runtime}{avg.\@ upper bound $b$}{avg.\@ runtime $t$ in sec.}
\addlogparetoplots{\letter}{Letter_HIGH__RESULTS_T10_S10.csv}{avg_classification_ratio}{avg_runtime}{avg.\@ classification ratio $r$}{avg.\@ runtime $t$ in sec.}
\addlogparetoplots{\aids}{AIDS__RESULTS_T10_S10.csv}{avg_ub}{avg_runtime}{avg.\@ upper bound $b$}{avg.\@ runtime $t$ in sec.}
\addlogparetoplots{\aids}{AIDS__RESULTS_T10_S10.csv}{avg_classification_ratio}{avg_runtime}{avg.\@ classification ratio $r$}{avg.\@ runtime $t$ in sec.}
\end{groupplot}
\node at ($(paretoplots c1r1.north) !.5! (paretoplots c2r1.north)$) [inner sep=0pt,anchor=south, yshift=3ex] {\pgfplotslegendfromname{grouplegend}};
\end{tikzpicture}

%% file: img/smolplots.tex
\begin{tikzpicture}
\begin{groupplot}
[
group style={group name=lineplots, group size=2 by 1, horizontal sep=1.5cm, vertical sep=1.25cm},
width=.5\linewidth,
height=.33\linewidth,
legend columns=5,
legend cell align=left,
legend style={align=left, draw=none, column sep=.5ex, font=\small}
]
\addlineplotsnomlsmol{\smolnoarg}{S-MOL-dev.csv}{num_labels}{deviation}{label alphabet size $|\Sigma_V|$}{avg.\@ deviation $d$ in \si{\percent}}
\addlineplotsnomlsmao{\sacyclicnoarg}{S-ACYCLIC-dev.csv}{num_labels}{deviation}{label alphabet size $|\Sigma_V|$}{avg.\@ deviation $d$ in \si{\percent}}
\end{groupplot}
\node at ($(lineplots c1r1.north) !.5! (lineplots c2r1.north)$) [inner sep=0pt,anchor=south, yshift=3ex] {\pgfplotslegendfromname{grouplegend}};
\end{tikzpicture}

%% file: conclusions.tex
\section{Conclusions and Future Work}\label{sec:conc}

In this paper, we formalized the paradigm \LSAPEGED for upper bounding \GED via transformations to \LSAPE and showed how to use machine learning in general in \SVM{s} in particular for this purpose. Moreover, we introduced rings, a new kind of local structures to be used by instantiations of \LSAPEGED, and presented the algorithms \RING and \RINGML that use rings to instantiate \LSAPEGED in a classical way (\RING) or via machine learning (\RINGML). 

Extensive experiments showed that, while existing instantiations of \LSAPEGED struggle with datasets where the information is mainly encoded in the graphs' topologies, \RING yields tight upper bounds also on such difficult instances. \RING hence closes the gap between existing instantiations of \LSAPE and accurate but very slow local search algorithms such as \IPFP. In other words, \RING is the first available \GED heuristics which allows to quickly compute reasonably tight upper bounds for \GED on difficult instances. In future work, we will seek to boost the performance of \RINGML by using automated feature selection techniques for weighing the importance of the employed local and global features.